\documentclass[a4paper]{article}
\usepackage{amssymb}
\usepackage{graphicx}

\newtheorem{definition}{Definition}
\newtheorem{theorem}{Theorem}
\newtheorem{lemma}{Lemma}
\newtheorem{corollary}{Corollary}
\newenvironment{proof}{\noindent \textbf{Proof}\ }{$\Box$ \bigskip}
\newtheorem{remark}{Remark}

\newenvironment{algorithm}
  {\begin{enumerate}\setlength{\itemsep}{-0.1em}}
  {\end{enumerate}}
\newlength{\aindent}
\newlength{\awidth}
\newcommand{\aitem}[2]{
  \addtolength{\awidth}{-#1\aindent}%
  \item \hspace{2em}\hspace{#1\aindent} \parbox[t]{\awidth}{\parindent=-3em #2}%
  \addtolength{\awidth}{#1\aindent}%
  }

\title{Efficient algorithms for enumerating maximal common subsequences of two strings}
\author{Miyuji Hirota\thanks{miyuji.hirota.p8@dc.tohoku.ac.jp} \and Yoshifumi Sakai\thanks{yoshifumi.sakai.c7@tohoku.ac.jp (Corresponding author)}}
\date{Graduate School of Agricultural Science, Tohoku University}

\begin{document}
\maketitle

\begin{abstract}
When searching extensively for significant common subsequences of two strings, one can consider maximal common subsequences (MCSs) as constituting the smallest set that encompasses all common subsequences, where a common subsequence is called maximal, if there exists no common subsequence that is longer than it and has it as a subsequence.
For any positive integer $n$, this article considers the problem of enumerating MCSs of two strings both of length $O(n)$.
Let a $(p, s, d)$-algorithm output all distinct MCSs each in $O(d)$ time after performing an $O(p)$-time preprocessing to construct a data structure of size $O(s)$.
We propose $(n^3, n^3, n)$-, $(n^2, n^2, n)$-, and $(n^2, n, n \log n)$-algorithms.
Although the $O(n^3, n^3, n)$-algorithm is inferior to the $(n^2, n^2, n)$-algorithm in terms of efficiency, the data structure created by it permits access to all distinct MCSs without explicit enumeration and is hence suitable for efficient exploration of certain special MCSs.
The $(n^2, n^2, n)$- and $(n^2, n, n \log n)$-algorithms are modifications of the $(\sigma n^2 \log n, n^2, \sigma n \log n)$-algorithm of [Conte et al., Algorithmica 84 (2022) 757--783], where $\sigma$ is the size of the alphabet over which the two strings are drawn.
\end{abstract}

\section{Introduction}\label{sec introduction}

To analyze sequential data consisting of characters, comparing two strings to search for significant patterns they share is a fundamental task.
One of the most classic and well-studied such patterns is the longest common subsequence (LCS).
The LCS of strings $X$ and $Y$ is defined as the longest string that commonly appears in $X$ and $Y$ as their subsequence.
Here, a subsequence of a string is obtained from the string by deleting any number of characters at any position not necessarily contiguous.
The LCS of $X$ and $Y$ is not necessarily unique; in general, there may exist many.
The LCS problem consists of finding an arbitrary one of the LCSs for given strings $X$ and $Y$.
As is well known, for any strings $X$ and $Y$ both of length $O(n)$, the dynamic programming (DP) algorithm solves the LCS problem in $O(n^2)$ time~\cite{WF}.
The time complexity of this algorithm is almost optimal in the sense that for any positive real number $\varepsilon$, there exists no algorithm that solves the LCS problem in $O(n^{2 - \varepsilon})$ time, unless the strong exponential time hypothesis (SETH) is false~\cite{ABW}.
The fastest algorithm known so far solves the LCS problem based on the four-Russians technique in $O(n / \log^2 n)$ time~\cite{MP}, where we adopt the unit-cost RAM model in this article.
For strings $X$ and $Y$ parameterized not only by $n$ but also by values such as the alphabet size, the LCS length, the number of matches, the number of the dominant matches, etc., algorithms specific to those parameters have also been proposed~\cite{Apo,AG,CP,GH,HS,IR,Mye,NKY,Ric}.
The space complexity of the DP algorithm of \cite{WF}, which is $O(n^2)$, can be reduced to $O(n)$ by the divide-and-conquer method without increasing the time complexity~\cite{Hir}.
Generalizing the number of strings to be compared from two to any, the LCS problem for multiple strings was shown to be NP-hard~\cite{Mai}.

When searching for meaningful common subsequences of strings $X$ and $Y$, if the conditions to be satisfied by the common subsequence to be found are known in advance, then simply finding an LCS does not necessarily yield the desired common subsequence.
As common subsequences to be found in particular such situations, conditional LCSs are considered.
For example, given a string $P$ as the required pattern to be taken into account in searching for the common subsequences, the constrained LCS (or SEQ-IC-LCS) problem~\cite{Tsa} consists of finding an arbitrary longest subsequence common to $X$ and $Y$ that has $P$ as its subsequence.
By modifying the DP algorithm of \cite{WF} for the (non-conditional) LCS problem so as to work on a three-dimensional DP table, this conditional LCS problem was shown to be solvable in $O(n^3)$ time~\cite{CSF+}.
Other conditional LCS problems with respect to $P$ that were shown to be solvable in $O(n^3)$ time or faster include the restricted LCS (or SEQ-EC-LCS) problem~\cite{CC,GHL+}, the STR-IC-LCS problem~\cite{CC,Deo}, and the STR-EC-LCS problem~\cite{WWW+}.

Suppose that we want to consider all common subsequences of strings $X$ and $Y$ to search for some important common structures shared by the strings.
For any common subsequence $Z$ of $X$ and $Y$, any subsequence $\tilde{Z}$ of $Z$ is also a common subsequence of $X$ and $Y$, so we will treat $Z$ as simultaneously representing $\tilde{Z}$.
For example, if $X = \mathtt{acbcded}$ and $Y = \mathtt{edeabcb}$, then we want to consider $\mathtt{a}$, $\mathtt{b}$, $\mathtt{c}$, $\mathtt{d}$, $\mathtt{e}$, $\mathtt{ab}$, $\mathtt{ac}$, $\mathtt{bc}$, $\mathtt{cb}$, $\mathtt{de}$, $\mathtt{ed}$, $\mathtt{abc}$, and $\mathtt{acb}$, but it is sufficient to explicitly consider only $\mathtt{de}$, $\mathtt{ed}$, $\mathtt{abc}$, and $\mathtt{acb}$.
Under this perspective, what common subsequences should we explicitly consider as the minimum set to represent all common subsequences?
If only all LCSs are considered explicitly, not all common subsequences are necessarily represented.
In other words, the requirement for maximum length is too restrictive.
For example, the previous concrete $X$ and $Y$ have $\mathtt{abc}$ and $\mathtt{acb}$ as LCSs, but neither represents $\mathtt{d}$, $\mathtt{e}$, $\mathtt{de}$, or $\mathtt{ed}$.
As a common subsequence that satisfies a loose alternative to the requirement of maximum length, let a \emph{maximal common subsequence} (an \emph{MCS}) be a common subsequence that is no longer a common subsequence no matter what character is inserted in any position.
From this definition, any common subsequence is represented by at least one MCS.
Conversely, any MCS is not represented by any common subsequence other than it.
Thus, our intended minimum set consists only of all MCSs.

As seen above, MCSs can be regarded as constituting the smallest set that represents all common subsequences.
Despite this useful feature, MCSs have not been studied very well at this time.
A few known results are as follows.
The shortest MCS problem consists of finding an arbitrary MCS of $X$ and $Y$ that has the least length.
This problem was shown to be solvable in $O(n^3)$ time and in $O(n^3)$ space~\cite{FIM}.
Given a common subsequence $P$ of $X$ and $Y$ arbitrarily, the constrained MCS problem consists of finding an arbitrary MCS of $X$ and $Y$ that has $P$ as its subsequence.
This problem can be solved in $O(n \log n)$ time and in $O(n)$ space\footnote{
In \cite{Sak} it is claimed that the MCS problem can be solved in $O(n \sqrt{\log n / \log \log n})$ time and $O(n)$ space by using the data structure of Beame and Fich~\cite{BFich} structure.
However, this is incorrect because it does not take into account the time to construct the data structure nor the space to store it.
The execution time and required space presented in the text are established by replacing their data structure with a naive data structure $\mathit{Index}_{\mathit{next/prev}}$, which is introduced in Section~\ref{sec pre}.
}
and whether a given common subsequence $Z$ of $X$ and $Y$ is maximal or not can be determined in $O(n)$ time \cite{Sak}.
Recently, Conte~et~al.~\cite{CGP+} showed that MCSs can be enumerated with a polynomial-time delay.
After performing an $O(\sigma n^2 \log n)$-time preprocessing to prepare a certain data structure of size $O(n^2)$, their algorithm outputs all distinct MCSs of $X$ and $Y$ each in $O(\sigma n \log n)$ time, where $\sigma$ is the number of characters in the alphabet.
With this result of Conte~et~al.~\cite{CGP+} we have for the first time a way to access the list of all MCSs.
The purpose of this article is to explore more efficient ways to enumerate MCSs with the goal of providing easier access to all MCSs.

\subsection{Our contribution}

To design algorithms for enumerating MCSs, we adopt two approaches, ``the all-in-one data structure approach'' and ``the prefix extension approach.''

In the all-in-one data structure approach, we design a directed acyclic graph (DAG), which we call the all-MCS graph, that represents all MCSs directly in the following sense.
This DAG has a single source vertex, having no incoming edge, and a single sink vertex, having no outgoing edge.
Furthermore, each path from the source vertex to the sink vertex represents a distinct MCS and vice versa, where the $k$th vertex in the path corresponds to the $k$th element of the MCS.
Thus, we can enumerate MCSs by enumerating paths from the source vertex to the sink vertex of the all-MCS graph.
We show that the all-MCS graph, satisfying the above conditions, exists as a DAG of size $O(n^3)$ and also show that this DAG can be constructed in $O(n^3)$ time from $X$ and $Y$. 
In enumeration using the all-MCS graph, the delay time is superior to the algorithm of Conte~et~al.~\cite{CGP+} but inferior in the preprocessing time and required space.
Perhaps the greatest strength of the all-MCS graph is that it allows access to all MCSs without explicitly enumerating them.
It is possible to utilize this characteristic to efficiently find certain special MCSs, a few examples of which are mentioned in the remarks.

The prefix extension approach is exactly the one adopted by the algorithm of Conte~et~al.~\cite{CGP+}.
In other words, we try to modify their algorithm to be more efficient.
Their approach is to build each MCS by repeatedly appending a valid character to the current prefix.
Since the valid characters to be appended to the prefix are determined by the support of the data structure, the design of it directly affects the efficiency of the algorithm.
We propose two algorithms by replacing the original data structure of Conte et al.~\cite{CGP+} with another.
One algorithm aims to minimize the delay time.
It outputs each MCS in $O(n)$ time after performing an $O(n^2)$-time preprocessing to prepare a data structure of size $O(n^2)$.
The other algorithm aims to minimize the required space to store the data structure.
It outputs each MCS in $O(n \log n)$ time after performing an $O(n^2)$-time preprocessing to prepare a data structure of size $O(n)$.
The efficiency of either algorithm is hence independent of the size $\sigma$ of the alphabet.
Furthermore, these algorithms successfully improve either the delay time of the Conte et al.~\cite{CGP+}'s algorithm by a factor of $\log n$ or the required space by a factor of $n$.
The efficiency of these proposed algorithms may appear to outperform the algorithm of Conte et al.~\cite{CGP+}.
However, since the efficiency of their algorithm can be evaluated using other parameters in addition to $n$ and $\sigma$, strictly speaking, the performance of our algorithms is not comparable to theirs.

This article is organized as follows.
Section~\ref{sec pre} defines our problem formally and introduces notations and terminology used in this article.
Section~\ref{sec all-in-one} proposes the all-in-one data structure approach algorithm by defining the all-MCS graph.
Section~\ref{sec ext} proposes the prefix extension approach algorithms by introducing the algorithm of Conte et al.~\cite{CGP+} as the basis for the modification in Section~\ref{sec Conte} and designing the data structures adopted by our algorithms in Sections~\ref{sec Enum221} and \ref{sec Enum211}.
Section~\ref{sec conc} concludes this article.

\section{Preliminaries}\label{sec pre}

For any sequences $S$ and $S'$, let $S \circ S'$ denote the sequence obtained by concatenating $S'$ after $S$.
For any sequence $S$, let $|S|$ denote the number of elements composing $S$.
For any index $k$ with $1 \leq k \leq |S|$, let $S[k]$ denote the $k$th element of $S$, so that $S = S[1] \circ S[2] \circ \cdots \circ S[|S|]$.
A \emph{subsequence} of $S$ is the sequence obtained from $S$ by deleting any number of elements at any position not necessarily contiguous, i.e., $S[k_1] \circ S[k_2] \circ \cdots \circ S[k_\ell]$ for some length $\ell$ with $0 \leq \ell \leq |S|$ and any $\ell$ indices $k_1,k_2,\dots,k_\ell$ with $1 \leq k_1 < k_2 < \cdots < k_\ell \leq |S|$.
We say that sequence $S$ \emph{contains} sequence $S'$, if $S'$ is a subsequence of $S$.
For any indices $k$ and $l$ with $1 \leq k \leq l \leq |S|$, let $S[k : l]$ denote the contiguous subsequence $S[k] \circ S[k + 1] \circ \cdots \circ S[l]$ of $S$.
For convenience, $S[k : k - 1]$ with $1 \leq k \leq |S| + 1$ denotes the empty contiguous subsequence of $S$.
Any contiguous subsequence $S[1 : k]$ with $0 \leq k \leq |S|$ is called a \emph{prefix} of $S$ and is denoted by $S \langle k]$.
Any contiguous subsequence $S[k : |S|]$ with $1 \leq k \leq |S| + 1$ is called a \emph{suffix} of $S$ and is denoted by $S[k \rangle$.
A \emph{string} is a sequence whose elements are characters in an alphabet.

Let $X$ and $Y$ be arbitrary strings over an alphabet $\Sigma$ of $\sigma$ characters.
Any string that both $X$ and $Y$ contain is called a \emph{common subsequence} of $X$ and $Y$.
We say that $X$ and $Y$ \emph{share} $Z$, if $Z$ is a common subsequece of $X$ and $Y$.
A \emph{maximal common subsequence} (an \emph{MCS}) of $X$ and $Y$ is a common subsequence of $X$ and $Y$ in which inserting any character no longer yields a common subsequence of $X$ and $Y$.
We consider the problem of enumerating MCSs of any strings $X$ and $Y$ both of length $O(n)$, where $n$ is an arbitrary positive integer.
Any algorithm that solves this problem would have to find each of all MCSs of $X$ and $Y$ exactly once within a certain delay time, perhaps after a certain preprocessing.
We call this problem the \emph{MCS enumeration problem}.
Let an $(f_{\mathrm{p}}(n, \sigma), f_{\mathrm{s}}(n, \sigma), f_{\mathrm{d}}(n, \sigma))$-\emph{algorithm} solve this problem, if it performs an $O(f_{\mathrm{p}}(n, \sigma))$-time preprocessing, uses $O(f_{\mathrm{s}}(n, \sigma))$ space, and  outputs all distinct MCSs of $X$ and $Y$ one by one each in $O(f_{\mathrm{d}}(n, \sigma))$ time.
We call $f_{\mathrm{p}}(n, \sigma)$, $f_{\mathrm{s}}(n, \sigma)$, and $f_{\mathrm{d}}(n, \sigma)$ the \emph{preprocessing-time, space, and delay-time complexities} of the algorithm, respectively.

In what follows, for convenience, we assume without loss of generality that $\Sigma = \{ 1,2,\dots,\sigma \}$, $X[1] = Y[1] = 1$, $X[|X|] = Y[|Y|] = \sigma$, both $X[2 : |X| -1]$ and $Y[2 : |Y| - 1]$ are strings over $\{ 2,3,\dots,\sigma - 1 \}$, and $\sigma \leq |X| + |Y| - 2 = O(n)$.
Note that $Z$ is an MCS of $X$ and $Y$ if and only if $Z[1] = 1$, $Z[|Z|] = \sigma$, and $Z[2 : |Z| - 1]$ is an MCS of $X[2 : |X| - 1]$ and $Y[2 : |Y| - 1]$.
Another assumption is that any sequence $S$ maintained by an algorithm is implemented as a one-dimensional array of $O(|S|)$ elements.
Therefore, any element $S[k]$ of $S$ with $1 \leq k \leq |S|$ can be accessed in $O(1)$ time, the first or last element of $S$ can be deleted in  $O(1)$ time, and any additional element can be appended or prepended to $S$ in $O(1)$ amortized time.

Below, we introduce notations and terminology that are used to design our algorithms.

For any string $W$ in $\{ X, Y \}$, any character $c$ with $1 \leq c \leq \sigma$, and any index $h$ with $1 \leq h \leq |W|$, let $\mathit{next}_W(c, h)$ ($\mathit{prev}_W(c, h)$) denote the least (resp. greatest) index $h'$ with $h < h' \leq |W|$ (resp. $1 \leq h' < h$) such that $W[h'] = c$, if any, or $|W| + 1$ (resp. $0$), otherwise. 
Let queries of any of these indices be called \emph{next/prev-queries}.
To support next/prev-queries, we consider the following two data structures.
One is $\mathit{Table}_{\mathit{next/prev}}$, which consists of pairs of sequences $\mathit{next}_W(c, 1) \circ \mathit{next}_W(c, 2) \circ \cdots \circ \mathit{next}_W(c, |W|)$ and $\mathit{prev}_W(c, 1) \circ \mathit{prev}_W(c, 2) \circ \cdots \circ \mathit{prev}_W(c, |W|)$ for all strings $W$ in $\{ X, Y \}$ and all characters $c$ with $1 \leq c \leq \sigma$.
This data structure is $O(n^2)$-time constructible, is of size $O(n^2)$, and supports $O(1)$-time next/prev-queries by working as the lookup table.
The other data structure is $\mathit{Index}_{next/prev}$, which consists of sequences $h_1 \circ h_2 \circ \cdots \circ h_\ell$ for all strings $W$ in $\{ X, Y \}$ and all characters $c$ with $1 \leq c \leq \sigma$, where $h_1 \circ h_2 \circ \cdots \circ h_\ell$ is the sequence of all indices $h$ such that $W[h] = c$ in ascending order.
This data structure is $O(n)$-time constructible, is of size $O(n)$, and supports $O(\log n)$-time next/prev-queries by performing a binary search on one of the sequences.
Although we consider only these simple data structures, they do not represent a bottleneck in any of the preprocessing-time, space, and delay-time complexities for the algorithms we propose.

Let a \emph{match} be a pair $(i, j)$ of indices with $1 \leq i \leq |X|$ and $1 \leq j \leq |Y|$ such that $X[i] = Y[j]$.
For any match $w$, we use $i_w$ and $j_w$ to denote the indices such that $w = (i_w, j_w)$ and use $c_w$ to denote the character common to $X[i_w]$ and $Y[j_w]$.
In addition, we use $d_w$ (resp. $a_w$) to denote the \emph{diagonal} (resp. \emph{anti-diagonal}) \emph{coordinate} $j_w - i_w$ (resp. $i_w + j_w$) of $w$ by considering it as a point on a two-dimensional grid.
For any matches $u$ and $v$, let $u < v$ (resp. $u \leq v$) mean that $i_u < i_v$ and $j_u < j_v$ (resp. $i_u \leq i_v$ and $j_u \leq j_v$).
Let $u \lneq v$ mean that $u \leq v$ and at least $i_u < i_v$ or $j_u < j_v$.
Furthermore, let $u \prec v$ mean that $u < v$ and there exists no match $w$ such that $u < w < v$.

For any character $c$ with $1 \leq c \leq \sigma$ and any match $w$ such that $X[i_w + 1 \rangle$ and $Y[j_w + 1 \rangle$ (resp. $X \langle i_w - 1]$ and $Y \langle j_w - 1]$) share $c$, let $\mathit{next}(c, w)$ (resp. $\mathit{prev}(c, w)$) denote the match $(\mathit{next}_X(c, i_w), \mathit{next}_Y(c, j_w))$ (resp. $(\mathit{prev}_X(c, i_w), \mathit{prev}_Y(c, j_w))$).
For any string $Z$ that $X$ and $Y$ share, let $\mathit{pref}(Z)$ (resp. $\mathit{suff}(Z)$) denote the match $(i, j)$ such that $X \langle i]$ (resp. $X[i \rangle$) is the shortest prefix (resp. suffix) of $X$ that contains $Z$ and $j$ satisfies the same condition as $i$ with respect to $Y$.
Hence, for any character $c$ and string $Z$ such that $X$ and $Y$ share $Z \circ c$ (resp. $c \circ Z$), $\mathit{pref}(Z \circ c) = \mathit{next}(c, \mathit{pref}(Z))$ (resp. $\mathit{suff}(c \circ Z) = \mathit{prev}(c, \mathit{suff}(Z))$).
Note that for any strings $Z'$ and $Z''$ both shared by $X$ and $Y$, $X$ and $Y$ share $Z' \circ Z''$ if and only if $\mathit{pref}(Z') < \mathit{suff}(Z'')$.
Let any match that is $\mathit{pref}(Z)$ (resp. $\mathit{suff}(Z)$) for some string $Z$ be called a \emph{pref-match} (resp. \emph{suff-match}).
Note that for any match $w$, $w$ is a pref-match (resp. suff-match) if and only if either $w = (1, 1)$ (resp. $w = (|X|, |Y|)$) or there exists a pref-match (resp. suff-match) $w'$ such that $w = \mathit{next}(c_w, w')$ (resp. $w = \mathit{prev}(c_w, w')$).

For any sequence $S$ of integers, let $\mathit{RMQ}_S$ denote the \emph{range minimum query} (\emph{RMQ}) data structure \cite{BF} for $S$, which can be constructed in $O(|S|)$ time from $S$ and supports $O(1)$-time queries of $\mathit{RMQ}_S(k' : k'')$ for any indices $k'$ and $k''$ with $1 \leq k' \leq k'' \leq |S|$, where $\mathit{RMQ}_S(k' : k'')$ is the greatest index $k$ with $k' \leq k \leq k''$ such that any element in $S[k' : k - 1]$ is greater than $S[k]$.
For convenience, we sometimes use $-S$ to denote the sequence $-S[1] \circ -S[2] \circ \cdots \circ -S[|S|]$, so that $\mathit{RMQ}_{-S}$ can be used to support \emph{range maximum queries} in the sense that $\mathit{RMQ}_{-S}(k' : k'')$ is the greatest index $k$ with $k' \leq k \leq k''$ such that any element in $S[k' : k - 1]$ is less than $S[k]$.

\section{All-in-one data structure approach algorithm}\label{sec all-in-one}

This section proposes an $(n^3, n^3, n)$-algorithm that solves the MCS enumeration problem. 

We design the proposed algorithm based on the following lemma, which redefines MCSs.

\begin{lemma}[\cite{Sak}]\label{lem MCS}
For any string $Z$ that $X$ and $Y$ share, $Z$ is an MCS of $X$ and $Y$ if and only if $\mathit{pref}(Z \langle k]) \prec \mathit{suff}(Z[k + 1 \rangle)$ for any index $k$ with $0 \leq k \leq |Z|$.
\end{lemma}

\begin{proof}
If there exists an index $k$ with $0 \leq k \leq |Z|$ such that $\mathit{pref}(Z \langle k]) \prec \mathit{suff}(Z[k + 1 \rangle)$ does not hold, then there exists a match $v$ such that $\mathit{pref}(Z \langle k]) < v < \mathit{suff}(Z[k + 1 \rangle)$, implying that $X$ and $Y$ share $Z \langle k] \circ c_v \circ Z[k + 1 \rangle$; otherwise, for any index $k$ with $1 \leq k \leq |Z|$ and any character $c$ with $1 \leq c \leq \sigma$, $X$ and $Y$ do not share $Z \langle k] \circ c \circ Z[k + 1 \rangle$, because there exists no match $v$ such that $\mathit{pref}(Z \langle k]) < v < \mathit{suff}(Z[k + 1 \rangle)$.
\end{proof}

The proposed algorithm uses a directed acyclic graph (DAG) such that each of certain paths represents a distinct MCS of $X$ and $Y$ and vice versa.
We call this DAG the \emph{all-MCS graph} and define it as follows (see also Figure~\ref{fig DAG}).

\begin{figure}[t]
\centering
\includegraphics[width=9.5cm]{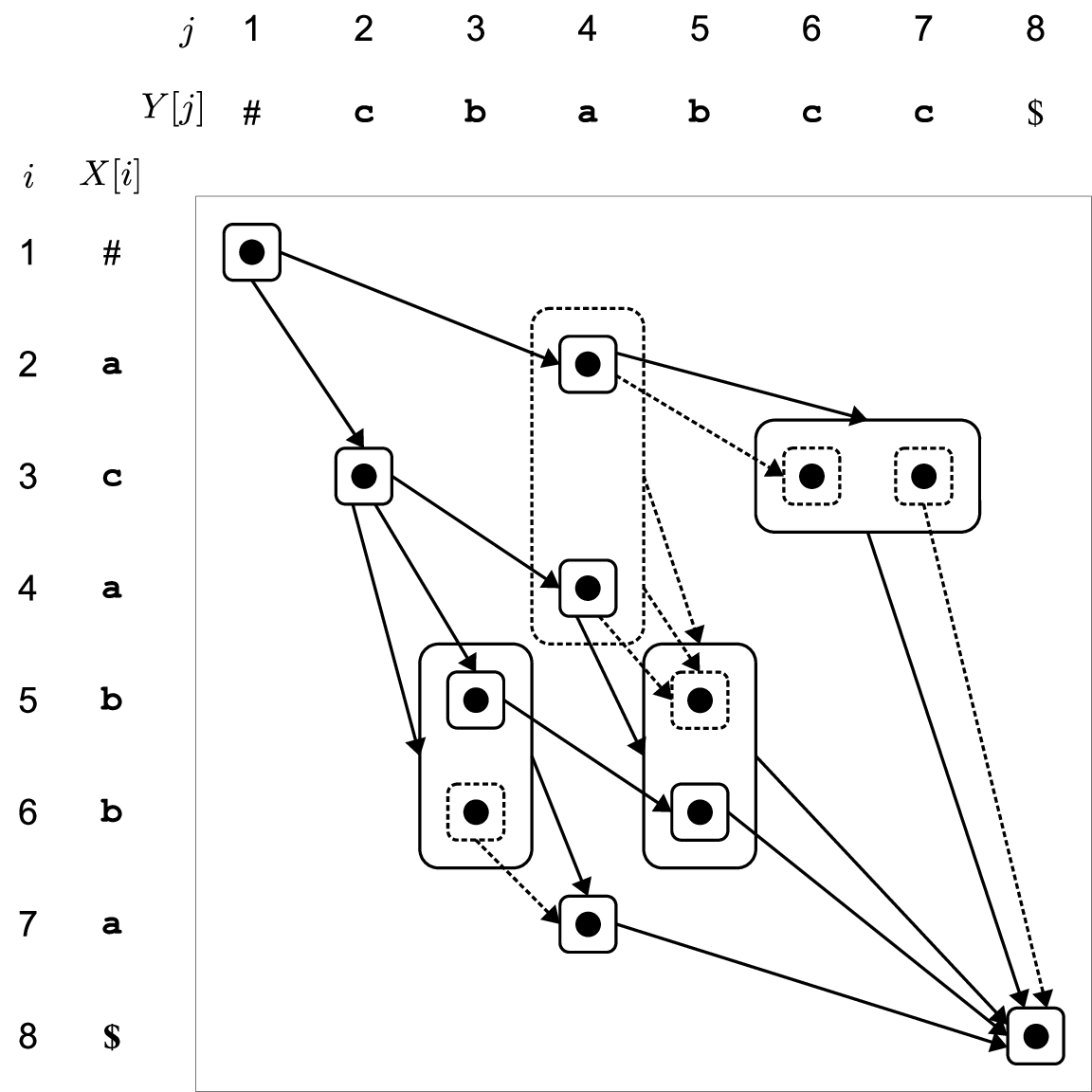}
\caption{
DAG $G$ for $X = \mathtt{\#acabba\$}$ and $Y = \mathtt{\#cbabcc\$}$ (with characters $1,2,3,4,5 \ (= \sigma)$ respectively represented by $\mathtt{\#,a,b,c,\$}$ and the concatenation operators $\circ$ omitted), where each match $w$ is indicated by a bullet at position $(i_w, j_w)$, each vertex $(w, w')$ is indicated by a rounded rectangle or square surrounding matches $w$ and $w'$, each edge is indicated by an arrow, and deleting all dotted vertices and edges yields the all-MCS graph
}
\label{fig DAG}
\end{figure}

\begin{definition}\label{def DAG}
Let $G$ be the DAG consisting of edges from vertex $(\mathit{prev}(c_u, v), u)$ to vertex $(v, \mathit{next}(c_v, u))$ for all pairs of matches $u$ and $v$ such that $u \prec v$.
The all-MCS graph is defined as the DAG that consists of all edges in $G$ through which a path from $((1, 1), (1, 1))$ to $((|X|, |Y|), (|X|, |Y|))$ in $G$ passes.
\end{definition}

\begin{lemma}\label{lem all-MCS}
The all-MCS graph has a path consisting of $\ell - 1$ edges from $(v_{k - 1}, u_k)$ to $(v_k, u_{k + 1})$ for all indices $k$ with $1 \leq k \leq \ell - 1$ such that $v_0 = u_1 = (1, 1)$ and $v_{\ell - 1} = u_\ell = (|X|, |Y|)$ if and only if $c_{u_1} \circ c_{u_2} \circ \cdots \circ c_{u_\ell}$ is an MCS of $X$ and $Y$.
\end{lemma}

\begin{proof}
To show the ``if'' part, let $Z$ be an arbitrary MCS of $X$ and $Y$.
For any index $k$ with $0 \leq k \leq |Z|$, let $u_k = \mathit{pref}(Z \langle k])$ and $v_k = \mathit{suff}(Z[k + 1 \rangle)$, so that $c_{u_1} \circ c_{u_2} \circ \cdots \circ c_{u_{|Z|}} = c_{v_0} \circ c_{v_1} \circ \cdots \circ c_{v_{|Z| - 1}} = Z$.
Since $v_0 = u_1 = (1, 1)$ and $v_{|Z| - 1} = u_{|Z|} = (|X|, |Y|)$ due to $Z[1] = 1$ and $Z[|Z|] = \sigma$, it suffices to show that for any index $k$ with $1 \leq k \leq |Z| - 1$, $G$ has an edge from $(v_{k - 1}, u_k)$ to $(v_k, u_{k + 1})$.
Since $Z[k \rangle = c_{u_k} \circ Z[k + 1 \rangle$ (resp. $Z \langle k + 1] = Z \langle k] \circ c_{v_k}$), $v_{k - 1} = \mathit{prev}(c_{u_k}, v_k)$ (resp. $u_{k + 1} = \mathit{next}(c_{v_k}, u_k)$).
Furthermore, $u_k \prec v_k$ due to Lemma~\ref{lem MCS}.
Thus, $G$ has an edge from $(v_{k - 1}, u_k)$ to $(v_k, u_{k + 1})$.

To show the ``only if'' part, consider an arbitrary path in $G$ that consists of $\ell - 1$ edges from $(v_{k - 1}, u_k)$ to $(v_k, u_{k + 1})$ for all indices $k$ with $1 \leq k \leq \ell - 1$ such that $v_0 = u_1 = (1, 1)$ and $v_{\ell - 1} = u_\ell = (|X|, |Y|)$ and let $Z = c_{u_1} \circ c_{u_2} \circ \cdots \circ c_{u_\ell} \ (= c_{v_0} \circ c_{v_1} \circ \cdots \circ c_{v_{\ell - 1}})$.
It can be verified by induction that for any index $k$ with $0 \leq k \leq |Z|$, $u_k = \mathit{pref}(Z \langle k])$ and $v_k = \mathit{suff}(Z[k + 1 \rangle)$.
Since $u_k \prec v_k$ for any index $k$ with $1 \leq k \leq |Z|$ due to definition of the all-MCS graph, it follows from Lemma~\ref{lem MCS} that $Z$ is an MCS of $X$ and $Y$.
\end{proof}

\begin{figure}
\centering
\begin{algorithm}
\aitem{1}{$(i, j) \leftarrow (|X|, j_u + 1)$;}
\aitem{1}{while $j \leq |Y|$,}
\aitem{2}{if $X[i] = Y[j]$, then}
\aitem{3}{output $(i, j)$;}
\aitem{2}{$(i, j) \leftarrow \left\{
\begin{array}{ll}
(\mathit{prev}_X(Y[j], i), j) & \mbox{if $\mathit{prev}_X(Y[j], i) > i_u$;} \\
(i, j + 1) & \mbox{otherwise.}
\end{array}
\right.$}
\end{algorithm}
\caption{
Procedure $\mathsf{FindEdges}(u)$}
\label{algo findEdges}
\end{figure}

The proposed algorithm constructs the all-MCS graph based on the following lemma.

\begin{lemma}\label{lem DAG}
The all-MCS graph can be constructed in $O(n^3)$ time and $O(n^3)$ space.
\end{lemma}

\begin{proof}
The all-MCS graph is obtained by constructing $G$ and modifying it as follows.

First of all, we construct $\mathit{Table}_{\mathit{next/prev}}$ of size $O(n^2)$ in $O(n^2)$ time so as to support $O(1)$-time next/prev-queries.
For any vertex $(v', u)$ in $G$, let $G(v', u)$ denote the sequence of all vertices $(v, u')$ such that $G$ has an edge from $(v', u)$ to $(v, u')$ in an arbitrary order, which represents the set of all outgoing edges from $(v', u)$.
Since $u \prec v$ and $v' = \mathit{prev}(c_u, v)$ for any such edge, $u \leq v'$ and at least $i_{v'} = i_u$ or $j_{v'} = j_u$.
Hence, if $d_{v'} \leq d_u$, then $v' = (i_u + (d_u - d_{v'}), j_u)$; otherwise, $v' = (i_u, j_u + (d_{v'} - d_u))$.
Based on this observation, we implement $G$ as a three-dimensional array of $O(n^3)$ elements $G[i, j, d]$ with $1 \leq i \leq |X|$, $1 \leq j \leq |Y|$, and $1 - |X| \leq d \leq |Y| - 1$, where $G[i, j, d]$ is $G(v', u)$, if $u = (i, j)$ and $d_{v'} = d$ for some vertex $(v', u)$ in $G$, or the empty sequence, otherwise.
This array can be constructed by initializing each element $G[i, j, d]$ to the empty sequence and appending $(v, \mathit{next}(c_v, u))$ to $G[i_u, j_u, d_{\mathit{prev}(c_u, v)}]$ for each pair of matches $u$ and $v$ such that $u \prec v$.
For any match $u$, Procedure $\mathsf{findEdges}(u)$ in Figure~\ref{algo findEdges} outputs all matches $v$ such that $u \prec v$ in $O(n)$ time because, by induction, just before any execution of line 3 of the procedure, for any index $i'$ with $1 \leq i' \leq |X|$, $u \prec (i', j)$ if and only if $i_u < i'$ and $X[i'] = Y[j]$. 
Since the number of matches $u$ is $O(n^2)$, $G$ can be constructed in $O(n^3)$ time and $O(n^3)$ space.

Next, we construct the graph, denoted by $G'$, that consists of all edges from $(v', u)$ to $(v, u')$ in $G$ such that $G$ has a path from $(1, 1)$ to $(v', u)$.
Using the breadth-first search algorithm, all such edges in $G$ can be determined in $O(n^3)$ time.
For any vertex $(v, u')$, let $G'(v, u')$ denote the sequence of all vertices $(v', u)$ such that $G'$ has an edge from $(v', u)$ to $(v, u')$ in an arbitrary order, which represents the set of all incoming edges to $(v, u')$.
We implement $G'$ as the array of elements $G'[i, j, d]$ in almost the same way as $G$.
The only difference is that $G'[i, j, d]$ is $G'(v, u')$, if $u' = (i, j)$ and $d_v = d$ for some vertex $(v, u')$ in $G'$.
The reason for this is to allow the breadth-first search algorithm to find all edges from $(v', u)$ to $(v, u')$ in $G'$ such that $G'$ has a path from $(v, u')$ to $((|X|, |Y|), (|X|, |Y|))$ in $O(n^3)$ time.
Obviously $G'$ can be constructed from $G$ in $O(n^3)$ time.

Finally, we construct the graph $G''$ that consists of all edges from $(v' u)$ to $(u, v')$ in $G'$ such that $G'$ has a path from $(v, u')$ to $((|X|, |Y|), (|X|, |Y|))$, which is hence the all-MCS graph.
For any vertex $(v', u)$ in $G''$, let $G''(v', u)$ denote the sequence of all vertices $(v, u')$ such that $G''$ has an edge from $(v', u)$ to $(v, u')$ in an arbitrary order.
We implement $G''$ as the array of elements $G''[i, j, d]$ in the same way as $G$.
Since the breadth-first search algorithm determines all edges from $(v', u)$ to $(v, u')$ in $G'$ such that $G'$ has a path from $(v, u')$ to $((|X|, |Y|), (|X|, |Y|))$ in $O(n^3)$ time, $G''$ can be constructed in $O(n^3)$ time.
\end{proof}

Let $G''$ be the array in the proof of Lemma~\ref{lem DAG}, which is our implementation of the all-MCS graph.
For any path $P$ in the all-MCS graph from $((1, 1), (1, 1))$ to $((|X|, |Y|), (|X|, |Y|))$, let $\mathit{id}(P)$ denote the sequence $r_1 \circ r_2 \circ \cdots \circ r_{\ell - 1}$ such that for any index $k$ with $1 \leq k \leq \ell - 1$, $(v_k, u_{k + 1})$ is the $r_k$th element of $G''(v_{k - 1}, u_k)$, where $P$ consists of $\ell - 1$ edges from $(v_{k - 1}, u_k)$ to $(v_k, u_{k + 1})$ for all indices $k$ with $1 \leq k \leq \ell - 1$ such that $v_0 = u_1 = (1, 1)$ and $v_{\ell - 1} = u_\ell = (|X|, |Y|)$.
Furthermore, let $\mathit{mcs}(P)$ denote the string $c_{u_1} \circ c_{u_2} \circ \cdots \circ c_{u_\ell}$.
The proposed algorithm outputs all distinct MCSs of $X$ and $Y$ based on the following lemma using $G''$.

\begin{lemma}\label{lem cubic}
If $G''$ in the proof of Lemma~\ref{lem DAG} is available, then all distinct MCSs of $X$ and $Y$ can be obtained one by one each in $O(n)$ time.
\end{lemma}

\begin{proof}
Array $G''$ allows us to obtain all distinct paths $P$ in the all-MCS graph from $((1, 1), (1, 1))$ to $((|X|, |Y|), (|X|. |Y|))$ in lexicographical order of $\mathit{id}(P)$, each in $O(n)$ time in a straightforward way.
String $\mathit{mcs}(P)$ for any such $P$ can be determined in $O(n)$ time.
Thus, the lemma follows from Lemma~\ref{lem all-MCS}.
\end{proof}

Let Algorithm $\mathsf{Enum331}$ be the algorithm that constructs our implementation $G''$ of the all-MCS graph based on Lemma~\ref{lem DAG} and outputs all distinct MCSs of $X$ and $Y$ according to Lemma~\ref{lem cubic}.
We immediately obtain the following theorem.

\begin{theorem}\label{theo Enum331} 
Algorithm $\mathsf{Enum331}$ is an $(n^3, n^3, n)$-algorithm that solves the MCS enumeration problem.
\end{theorem}

\begin{remark}
The delay time achieved by Algorithm $\mathsf{Enum331}$ to output each MCS $Z$ is $O(|Z|)$ rather than $O(n)$.
In the next section, we propose an $(n^2, n^2, n)$-algorithm, which performs more efficiently with respect to preprocessing time and required space than Algorithm $\mathsf{Enum331}$ but does not necessarily output each MCS $Z$ in $O(|Z|)$ time.
\end{remark}

\begin{remark}
The all-MCS graph constructed by Algorithm $\mathsf{Enum331}$ can be used not only for enumerating (non-conditional) MCSs but also for finding one of certain particular MCSs or enumerating them.
Such particular MCSs include, for example, \emph{quasi-LCSs} and \emph{most stable MCSs}.
A quasi-LCS is defined as one of the longest MCSs that is not an LCS.
Since no quasi-LCS can be obtained only by deleting characters from any LCS, we can think of quasi-LCSs as alternatives to LCSs in searching for meaningful common subsequences.
A most stable MCS is defined as an MCS $Z$ that has the greatest number of indices $k$ with $1 \leq k \leq |Z|$ such that $\mathit{pref}(Z \langle k]) = \mathit{suff}(Z[k \rangle)$.
LCSs may have only few characters whose positions in the strings are uniquely determined while most stable MCSs have a maximum number of such characters.
Therefore, most stable MCSs could be used to find positional correspondences between the strings.
For example, if we consider the same $X$ and $Y$ as in Figure~\ref{fig DAG}, then $\mathtt{\#ac\$}$ (resp. $\mathtt{\#cbb\$}$) is the only quasi-LCS (resp. most stable MCS) that $X$ and $Y$ have.
It is easy to design a DP algorithm that finds a quasi-LCS (resp. most stable MCS) in $O(n^3)$ time.
The DP table constructed by the algorithm consists of the maximum and second maximum numbers of vertices (resp. the maximum number of vertices $(w, w')$ with $w = w'$) through which a path from $((1, 1), (1, 1))$ to $(v', u)$ passes for all vertices $(v', u)$ in the all-MCS graph.
Since each traceback path corresponds to a distinct quasi-LCS (resp. most stable MCS) and vice versa, this DP table can also be used for enumeration.
\end{remark}

\section{Prefix extension approach algorithms}\label{sec ext}

This section modifies the $(\sigma n^2 \log n, n^2, \sigma n \log n)$-algorithm of Conte~et~al.~\cite{CGP+} for the MCS enumeration problem to obtain $(n^2, n^2, n)$- and $(n^2, n, n \log n)$-algorithms.

\subsection{Conte~et~al.~\cite{CGP+}'s prefix-extensible character test}\label{sec Conte}

We first introduce the approach adopted by Conte~et~al.~\cite{CGP+} to solve the MCS enumeration problem in our terminology.

Let any string $Z'$ that is a prefix of some MCS of $X$ and $Y$ be called an \emph{MCS-prefix}.
For any MCS-prefix $Z'$, let a character $c$ be $Z'$-\emph{extensible}, if $Z' \circ c$ is also an MCS-prefix.
If $Z'$-extensible characters can be determined somehow, then we can find all distinct MCSs of $X$ and $Y$ one by one in lexicographical order by executing Algorithm $\mathsf{Basic}$ in Figure~\ref{algo basic} (see also Figure~\ref{fig trie}).
To adopt this straightforward algorithmic approach, Conte~et~al.~\cite{CGP+} developed a clever way to test which characters are $Z'$-extensible.

\begin{figure}
\centering
\begin{algorithm}
\aitem{1}{Let $Z'$ be the string consisting only of a single character $1$;}
\aitem{1}{while $Z'$ is nonempty,}
\aitem{2}{if $Z'[|Z'|] \neq \sigma$, then}
\aitem{3}{append the least $Z'$-extensible character to $Z'$;}
\aitem{2}{otherwise,}
\aitem{3}{output $Z'$;}
\aitem{3}{while $Z'[|Z'|]$ is the greatest $Z' \langle |Z'| - 1]$-extensible character,}
\aitem{4}{delete the last element from $Z'$;}
\aitem{3}{if $Z'$ is nonempty, then}
\aitem{4}{replace the last element of $Z'$ with the least $Z \langle |Z'| - 1]$-extensible character that is greater than $Z'[|Z'|]$.}
\end{algorithm}
\caption{
Algorithm $\mathsf{Basic}$}
\label{algo basic}
\end{figure}

\begin{figure}
\centering
\includegraphics[width=10cm]{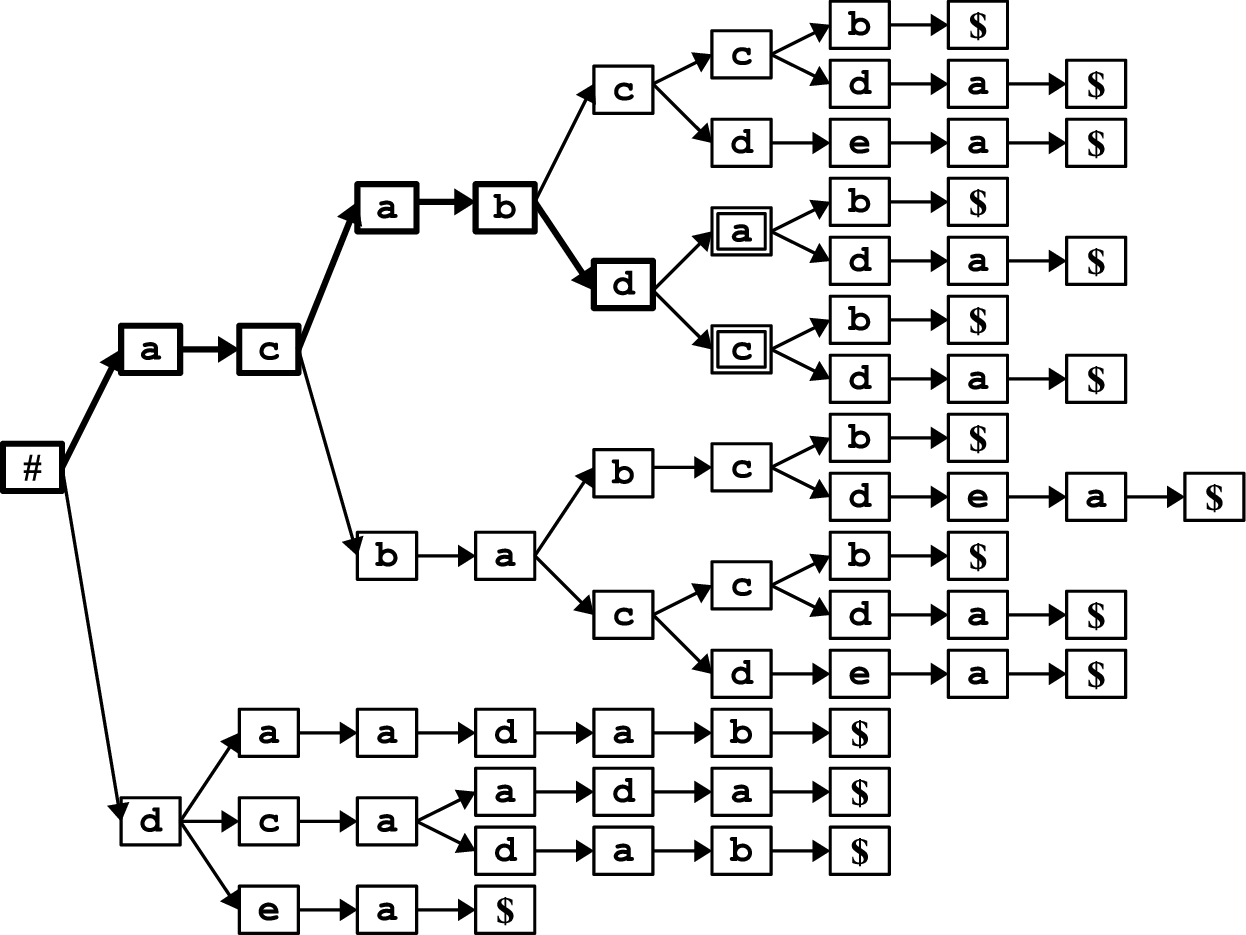}
\caption{
Trie representing all MCSs of $X = \mathtt{\#dacabdacbcbdea\$}$ and $Y = \mathtt{\#acbabcdecaadab\$}$ (with characters $1,2,3,4,5,6,7 \ (= \sigma)$ respectively represented by $\mathtt{\#,a,b,c,d,e,\$}$ and the concatenation operators $\circ$ omitted) in lexicographical order, on which Algorithm $\mathsf{Basic}$ in Figure~\ref{algo basic} performs a preorder tree walk to enumerate them, where, for example, the path indicated by thick nodes and edges represents an MCS-prefix $Z' = \mathtt{\#acabd}$, which is obtained by line 10 of the algorithm after $\mathtt{\#acabcdea\$}$ is output as the third MCS in the enumeration, and the $Z'$-extensible characters, which are $\mathtt{a}$ and $\mathtt{c}$, are indicated by double-edged nodes
}
\label{fig trie}
\end{figure}

\begin{figure}[t]
\centering
\includegraphics[width=10cm]{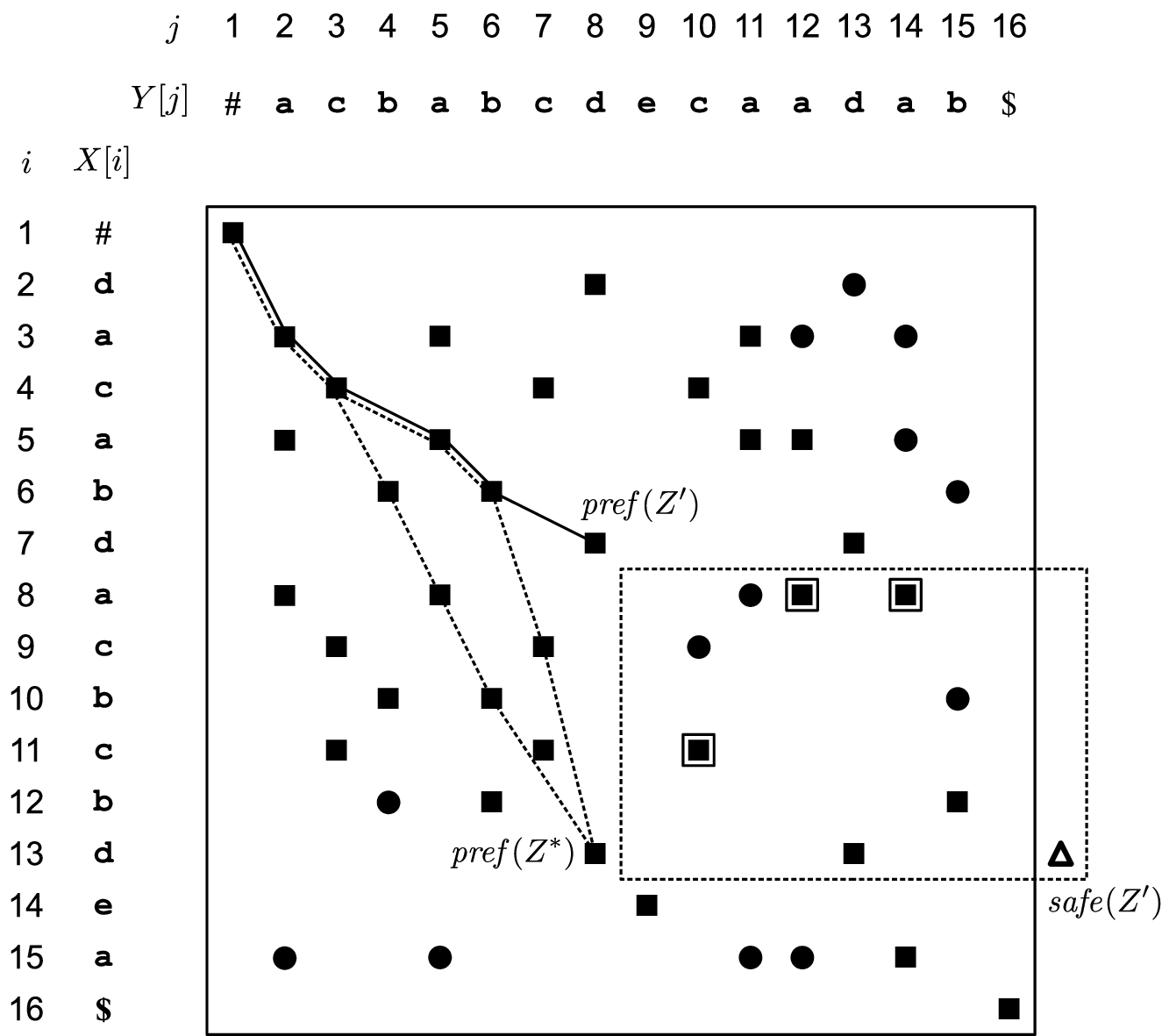}
\caption{
Suff-matches (indicated by square bullets), the other matches (indicated by circular bullets), and the region of matches $v$ such that $\mathit{pref}(Z') < v \leq \mathit{safe}(Z')$ (surrounded by a dotted rectangle) for the same $X$, $Y$, and $Z'$ as Figure~\ref{fig trie}, where $Z^*$ is any of $\mathtt{\#acabcd}$ and $\mathtt{\#acbabd}$, which are the single-character insertion derivatives of $Z'$ such that $i_{\mathit{safe}(Z')} = i_{\mathit{pref}(Z^*)}$, the triangular open bullet indicates $\mathit{safe}(Z')$, which is a virtual match in this concrete example, and the bullets indicating suff-matches $v$ with $\mathit{pref}(Z') \prec v \leq \mathit{safe}(Z')$ are double-edged 
}
\label{fig ext}
\end{figure}

For any string $Z$ that $X$ and $Y$ share, let any string $Z \langle k] \circ c \circ Z[k + 1 \rangle$ with $0 \leq k \leq |Z|$ and $1 \leq c \leq \sigma$ that $X$ and $Y$ share be called a \emph{single-character insertion derivative} of $Z$, so that $Z$ is an MCS of $X$ and $Y$ if and only if there exists no single-character insertion derivative of $Z$.
For any MCS-prefix $Z'$, let $\mathit{safe}(Z')$ denote the (possibly virtual) match $(i, j)$, where $i$ (resp. $j$) is the least possible index such that $\mathit{pref}(Z^*) = (i, j_{\mathit{pref}(Z')})$ (resp. $\mathit{pref}(Z^*) = (i_{\mathit{pref}(Z')}, j)$) for some single-character insertion derivative $Z^*$ of $Z'$, if any, or $|X| + 1$ (resp. $|Y| + 1$), otherwise.
Although $\mathit{safe}(Z')$ is not a match when $i_{\mathit{safe}(Z')} = |X| + 1$ or $j_{\mathit{safe}(Z')} = |Y| + 1$, we treat it as a virtual match to allow us to use notations $w < \mathit{safe}(Z')$ and $w \leq \mathit{safe}(Z')$ for any match $w$ (see Figure~\ref{fig ext}).
Conte et~al.~\cite{CGP+} revealed which characters are $Z'$-extensible by their relationship to $\mathit{safe}(Z')$.

\begin{lemma}[\cite{CGP+}]\label{lem ext}
For any MCS-prefix $Z'$ and any character $c$, $c$ is $Z'$-extensible if and only if there exists a suff-match $v$ such that $\mathit{pref}(Z') \prec v \leq \mathit{safe}(Z')$ and $c_v = c$.
\end{lemma}

\begin{proof}
We only consider the case where there exists a match $v$ such that $\mathit{pref}(Z') < v$ and $c_v = c$ because otherwise the lemma holds.
Let $v$ be an arbitrary suff-match with $\mathit{pref}(Z') < v$ and $c_v = c$ such that there exists no suff-match $v'$ with $\mathit{pref}(Z') < v' \lneq v$ and $c_{v'} = c$.
Let $Z''$ be an arbitrary longest string such that $\mathit{suff}(Z'') = v$.

To prove the ``if'' part of the lemma, suppose that $\mathit{pref}(Z') \prec \mathit{suff}(Z'') \leq \mathit{safe}(Z')$.
For any single-character insertion derivative $Z^{**}$ of $Z''$, if $\mathit{pref}(Z') < \mathit{suff}(Z^{**})$, then $\mathit{suff}(Z^{**}) < \mathit{suff}(Z'')$ due to definition of $v$ and $Z''$, which contradicts that $\mathit{pref}(Z') \prec \mathit{suff}(Z'')$.
On the other hand, for any single-character insertion derivative $Z^*$ of $Z'$, $\mathit{pref}(Z') \lneq \mathit{pref}(Z^*)$ because $Z'$ is an MCS-prefix.
This implies from definition of $\mathit{safe}(Z')$ that if $\mathit{pref}(Z^*) < \mathit{safe}(Z')$, then $\mathit{pref}(Z') < \mathit{pref}(Z^*)$.
Hence, from $\mathit{suff}(Z'') \leq \mathit{safe}(Z')$, $\mathit{pref}(Z^*) < \mathit{suff}(Z'')$ contradicts that $\mathit{pref}(Z') \prec \mathit{suff}(Z'')$.
Thus, there exists no single-character insertion derivative of $Z' \circ Z''$.

To prove the ``only if'' part of the lemma, suppose that there exists no single-character insertion derivative of $Z' \circ Z''$.
This immediately implies that $\mathit{pref}(Z') \prec \mathit{suff}(Z'')$.
Furthermore, $\mathit{suff}(Z'') \leq \mathit{safe}(Z')$ because otherwise it follows from definition of $\mathit{safe}(Z')$ that there exists a single-character insertion derivative $Z^*$ of $Z'$ such that $X$ and $Y$ share $Z^* \circ Z''$, a contradiction.
\end{proof}

Conte et~al.~\cite{CGP+} also gave an inductive method for updating $\mathit{safe}(Z')$ to $\mathit{safe}(Z' \circ c)$ for any $Z'$-extensible character $c$ as follows.

\begin{lemma}[\cite{CGP+}]\label{lem safe}
For any MCS-prefix $Z'$ and any $Z'$-extensible character $c$, $i_{\mathit{safe}(Z' \circ c)}$ is the minimum of the following at most three indices.
One is $|X| + 1$.
Another is the minimum of $\mathit{next}_X(c, \mathit{next}_X(Y[j], i_{\mathit{pref}(Z')}))$ over all indices $j$ with $j_{\mathit{pref}(Z')} < j < j_{\mathit{pref}(Z' \circ c)}$ such that $\mathit{next}_X(Y[j], i_{\mathit{pref}(Z')}) \leq |X|$, if any.
The other is $\mathit{next}_X(c, i_{\mathit{safe}(Z')})$, if $i_{\mathit{safe}(Z')} \leq |X|$.
Index $j_{\mathit{safe}(Z' \circ c)}$ can be determined analogously by exchanging the roles of $X$ and $Y$.
\end{lemma}

\begin{proof}
By symmetry, we show only that the condition of $i_{\mathit{safe}(Z' \circ c)}$ in the lemma holds.
Let $\mathcal{Z}^*$ be the set of all single-character insertion derivatives $Z^*$ of $Z'$ such that $j_{\mathit{pref}(Z^* \circ c)} = j_{\mathit{pref}(Z' \circ c)}$.
If $\mathcal{Z}^*$ is empty, then $i_{\mathit{safe}(Z' \circ c)} = |X| + 1$; otherwise, $i_{\mathit{safe}(Z' \circ c)}$ is the minimum of $\mathit{next}_X(c, i_{\mathit{pref}(Z^*)})$ over all strings $Z^*$ in $\mathcal{Z}^*$.
For any $Z^*$ in $\mathcal{Z}^*$, $j_{\mathit{pref}(Z^*)} \geq j_{\mathit{pref}(Z')}$ and if $j_{\mathit{pref}(Z^*)} > j_{\mathit{pref}(Z')}$, then $Z^* = Z' \circ Y[j]$ for some index $j$ with $j_{\mathit{pref}(Z')} < j < j_{\mathit{pref}(Z' \circ c)}$ such that $\mathit{next}_X(Y[j], i_{\mathit{pref}(Z')}) \leq |X|$.
The minimum of $\mathit{next}_X(c, i_{\mathit{pref}(Z^*)})$ over all strings $Z^*$ in $\mathcal{Z}^*$ such that $j_{\mathit{pref}(Z^*)} > j_{\mathit{pref}(Z')}$ is the minimum of $\mathit{next}_X(c, \mathit{next}_X(Y[j], i_{\mathit{pref}(Z')}))$ over all indices $j$ with $j_{\mathit{pref}(Z')} < j < j_{\mathit{pref}(Z' \circ c)}$.
On the other hand, it follows from definition of $\mathit{safe}(Z')$ that if $i_{\mathit{safe}(Z')} \leq |X|$, then the minimum of $\mathit{next}_X(c, i_{\mathit{pref}(Z^*)})$ over all strings $Z^*$ in $\mathcal{Z}^*$ such that $j_{\mathit{pref}(Z^*)} = j_{\mathit{pref}(Z')}$ is $\mathit{next}_X(c, i_{\mathit{safe}(Z')})$; otherwise, there exists no such $Z^*$.
\end{proof}

Based on Lemma~\ref{lem ext}, for any MCS-prefix $Z'$, let any suff-match $v$ such that $\mathit{pref}(Z') \prec v \leq \mathit{safe}(Z')$ be called a \emph{witness} (or \emph{$c_v$-witness}) \emph{of $Z'$-extensibility}.
In addition, let any witness $v$ of $Z'$-extensibility be called \emph{prominent}, if there exists no witness $v'$ of $Z'$-extensibility such that $v' \lneq v$.
For example, in the case of Figure~\ref{fig ext}, there exist three witnesses, $(8, 12)$, $(8, 14)$, and $(11, 10)$, of $Z'$-extensibility and only $(8, 12)$ and $(11, 10)$ are prominent due to $(8, 12) \lneq (8, 14)$.
Let $\mathit{ext}(Z')$ denote the sequence of all prominent witnesses of $Z'$-extensibility in an arbitrary order.
The reason why we introduce this sequence is as follows.
For any character $c$ with $1 \leq c \leq \sigma$ and any $c$-witness $v$ of $Z'$-extensibility, at least $i_v = i_{\mathit{pref}(Z' \circ c)}$ or $j_v = j_{\mathit{pref}(Z' \circ c)}$ due to $\mathit{pref}(Z') \prec v$.
This implies that there exist at most two prominent $c$-witnesses of $Z'$-extensibility.
Furthermore, there exists no prominent $c$-witness of $Z'$-extensibility if and only if there exists no $c$-witness of $Z'$-extensibility.
This implies from Lemma~\ref{lem ext} that $c$ is $Z'$-extensible if and only if there exists a prominent $c$-witness of $Z'$-extensibility. 
Thus, to implement Algorithm $\mathsf{Basic}$ so as to run efficiently, we can concentrate on designing an efficient data structure that supports queries of $\mathit{ext}(Z')$ for any MCS-prefix $Z'$. 
This is because if $\mathit{ext}(Z')$ is available, then all $Z'$-extensible characters can be determined in time linear in the number of them.

The $(\sigma n^2 \log n, n^2, \sigma n \log n)$-algorithm of Conte~et~al.~\cite{CGP+} can be thought of as adopting an $O(n^2)$-space data structure that supports $O(\sigma \log n)$-time queries of $\mathit{ext}(Z')$, which stores all suff-matches to search for a candidate of each of two possible prominent $c$-witnesses of $Z'$-extensibility in $O(\log n)$ time.
To propose $(n^2, n^2, n)$- and $O(n^2, n, n \log n)$-algorithms, we develop different data structures to efficiently support queries of $\mathit{ext}(Z')$.

\begin{remark}
The delay time achieved by the algorithm of Conte~et~al.~\cite{CGP+} to output each MCS $Z$ is $O(\sigma |Z| \log n)$ rather than $O(\sigma n \log n)$.
The space complexity of the algorithm is $O(\sigma n/\log n + R)$ rather than $O(n^2)$, where $R$ is the number of suff-matches.
The preprocessing-time complexity is $O(\sigma n + \sigma M \log n)$ rather than $O(\sigma n^2 \log n)$, where $M$ is the number of matches.
In contrast, our algorithms proposed in the subsequent sections are designed so that the efficiency depends only on $n$, independent of $\sigma$, $|Z|$, $R$, or $M$.
Therefore, strictly speaking, it does not make sense to simply compare the efficiency of their algorithm and ours only with respect to $n$.
\end{remark}

\subsection{$O(n^2, n^2, n)$-algorithm}\label{sec Enum221}

As an implementation of Algorithm $\mathsf{Basic}$ based on Conte~et~al.~\cite{CGP+}'s prefix-extensible character test (Lemmas~\ref{lem ext} and \ref{lem safe}), we propose an $O(n^2, n^2, n)$-algorithm that solves the MCS enumeration problem, which we denote Algorithm $\mathsf{Enum221}$.

Algorithm $\mathsf{Enum221}$ utilizes an $O(n^2)$-time constructible data structure $\mathit{D221}$ supporting $O(|\mathit{ext}(Z')|)$-time queries of $\mathit{ext}(Z')$ for any MCS-prefix $Z'$.
Before proceeding to the design of $\mathit{D221}$, we observe below why this data structure works to output all distinct MCSs of $X$ and $Y$ one by one each in $O(n)$ time.

To simulate Algorithm $\mathsf{Basic}$, Algorithm $\mathsf{Enum221}$ uses $\mathit{Table}_{\mathit{next/prev}}$ to support $O(1)$-time next/prev-queries and maintains $\mathit{pref}(Z' \langle k])$, $\mathit{safe}(Z' \langle k])$, and $\mathit{alt}(Z' \langle k])$ for all indices $k$ with $1 \leq k \leq |Z'|$ as well as $Z'$, where $\mathit{alt}(Z')$ denotes a sequence of all $Z' \langle |Z'| - 1]$-extensible characters that are greater than $Z'[|Z'|]$ in an arbitrary order.
It is easy to simulate line 1 of Algorithm $\mathsf{Basic}$ because $\mathit{pref}(Z') = (1, 1)$, $\mathit{safe}(Z') = (|X| + 1, |Y| + 1)$, and $\mathit{alt}(Z')$ is the empty sequence.
Line 4 is simulated by using $\mathit{D221}$ to obtain $\mathit{ext}(Z')$ in $O(|\mathit{ext}(Z')|)$ time and scanning $\mathit{ext}(Z')$ to both determine the least $Z'$-extensible character $c$ and construct $\mathit{alt}(Z' \circ c)$ in $O(|\mathit{ext}(Z')|)$ time.
Furthermore, $\mathit{Table}_{\mathit{next/prev}}$ is used to obtain $\mathit{pref}(Z' \circ c) = \mathit{next}(c, \mathit{pref}(Z'))$ in $O(1)$ time and determine $\mathit{safe}(Z' \circ c)$ based on Lemma~\ref{lem safe} in $O(a_{\mathit{pref}(Z' \circ c)} - a_{\mathit{pref}(Z')})$ time, where we note that $a_{\mathit{pref}(Z' \circ c)} - a_{\mathit{pref}(Z')} = (i_{\mathit{pref}(Z' \circ c)} - i_{\mathit{pref}(Z')}) + (j_{\mathit{pref}(Z' \circ c)} - j_{\mathit{pref}(Z')})$.
The condition in line 7 holds if and only if $\mathit{alt}(Z')$ is empty.
Line 10 is simulated by scanning $\mathit{alt}(Z')$ to decompose it into the least $Z'$-extensible character $c$ in it and $\mathit{alt}(Z' \langle |Z'| - 1] \circ c)$ in $O(|\mathit{ext}(Z' \langle |Z'| - 1])|)$ time and using $\mathit{Table}_{\mathit{next/prev}}$ to obtain $\mathit{pref}(Z' \langle |Z'| - 1] \circ c) = \mathit{next}(c, \mathit{pref}(Z' \langle |Z'| - 1]))$ in $O(1)$ time and determine $\mathit{safe}(Z' \langle |Z'| - 1] \circ c)$ based on Lemma~\ref{lem safe} in $O(a_{\mathit{pref}(Z' \langle |Z'| - 1] \circ c)} - a_{\mathit{pref}(Z' \langle |Z'| - 1])})$ time.

From the above implementation of Algorithm $\mathsf{Basic}$ adopted by Algorithm $\mathsf{Enum221}$, each MCS $Z'$ of $X$ and $Y$ is obtained in time linear in the sum of $(a_{\mathit{pref}(Z' \langle k + 1])} - a_{\mathit{pref}(Z' \langle k])}) + |\mathit{ext}(Z' \langle k])|$ over all indices $k$ with $1 \leq k \leq |Z'| - 1$.
Since the sum of $a_{\mathit{pref}(Z' \langle k + 1])} - a_{\mathit{pref}(Z' \langle k])}$ over all such indices $k$ is equal to $a_{Z'} - a_{Z'[1]} = a_{(|X|, |Y|)} - a_{(1, 1)} = |X| + |Y| - 2 = O(n)$, if $|\mathit{ext}(Z' \langle k])| = O(a_{\mathit{pref}(Z' \langle k + 1])} - a_{\mathit{pref}(Z' \langle k])})$, then each $Z'$ is output in $O(n)$ time.
The following lemma ensures that this delay-time complexity is actually established.

\begin{lemma}\label{lem Enum221}
For any MCS-prefix $Z'$ and any $Z'$-extensible character $c$, the number of $Z'$-extensible characters is less than $a_{\mathit{pref}(Z' \circ c)} - a_{\mathit{pref}(Z')}$, implying that $|\mathit{ext}(Z')| = O(a_{\mathit{pref}(Z' \circ c)} - a_{\mathit{pref}(Z')})$.
\end{lemma}

\begin{proof}
Let $(i_1, j_1),(i_2, j_2),\dots,(i_s, j_s)$ be matches $\mathit{pref}(Z' \circ c)$ for all distinct $Z'$-extensible characters $c$ in ascending order of $j_{\mathit{pref}(Z' \circ c)}$ (hence also in descending order of $i_{\mathit{pref}(Z' \circ c)}$ due to $\mathit{pref}(Z') \prec \mathit{pref}(Z' \circ c)$).
For any index $r$ with $1 \leq r \leq s$, both $i_{\mathit{pref}(Z')} < i_s < i_{s - 1} < \cdots < i_r$ and $j_{\mathit{pref}(Z')} < j_1 < j_2 < \cdots < j_r$.
This implies that $a_{(i_r, j_r)} - a_{\mathit{pref}(Z')} = (i_r - i_{\mathit{pref}(Z')}) + (j_r - j_{\mathit{pref}(Z')}) \geq ((s + 1) - r)) + r = s + 1$.
\end{proof}

Lemma~\ref{lem Enum221} immediately yields its corollary as follows.

\begin{corollary}\label{cor Enum221}
If $\mathit{D221}$ supports $O(|\mathit{ext}(Z')|)$-time queries of $\mathit{ext}(Z')$ for any MCS-prefix $Z'$, then Algorithm $\mathsf{Enum221}$ uses $O(n)$ space, excluding space for storing $\mathit{Table}_{\mathit{next/prev}}$ and $\mathit{D221}$, to output all distinct MCSs of $X$ and $Y$ one by one each in $O(n)$ time.
\end{corollary}

To achieve $O(|\mathit{ext}(Z')|)$-time queries of $\mathit{ext}(Z')$ for any MCS-prefix $Z'$, we define data structure $\mathit{D221}$ as follows.

\begin{definition}\label{def D221}
For any index $i$ with $1 \leq i \leq |X|$, let $\mathit{Suff}_i$ denote the sequence of $|Y|$ indices such that for any index $j$ with $1 \leq j \leq |Y|$, $\mathit{Suff}_i[j]$ is the least index $i'$ with $i \leq i' \leq |X|$ such that $(i', j)$ is a suff-match, if any, or $|X| + 1$, otherwise (see Figure~\ref{fig D221}).
Let $\mathit{D221}$ consist of sequences $\mathit{Suff}_i$ and the RMQ data structures $\mathit{RMQ}_{\mathit{Suff}_i}$ for all indices $i$ with $1 \leq i \leq |X|$, which is hence of size $O(n^2)$.
\end{definition}

\begin{figure}[t]
\centering
\includegraphics[width=10cm]{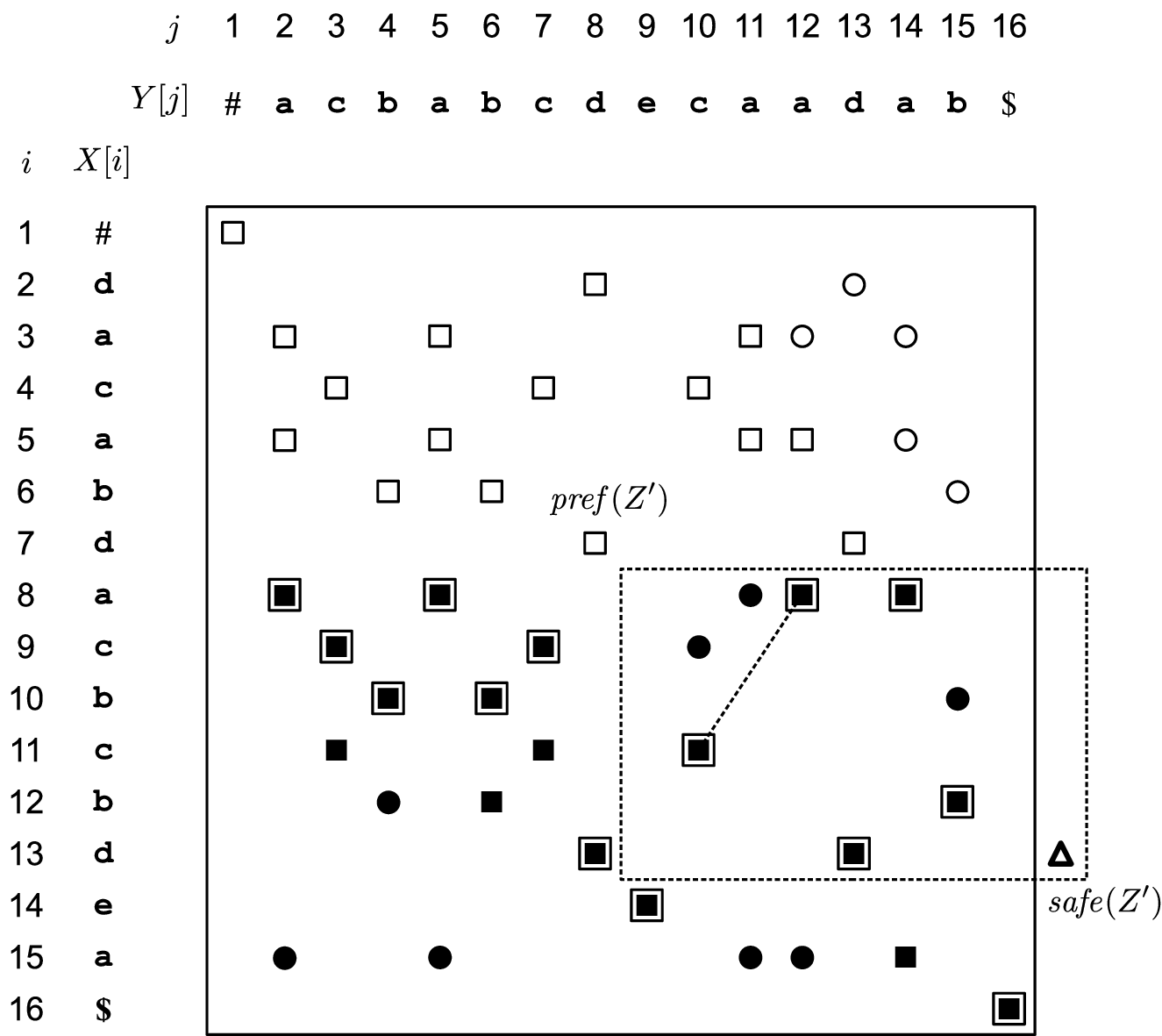}
\caption{
Sequence $\mathit{Suff}_i = 17 \circ 8 \circ 9 \circ 10 \circ 8 \circ 10 \circ 9 \circ 13 \circ 14 \circ 11 \circ 17 \circ 8 \circ 13 \circ 8 \circ 12 \circ 16$ with $i = i_{\mathit{pref}(Z')} + 1 \ (= 8)$ and $\mathit{ext}(Z') = (8, 12) \circ (11, 10)$ for the same $X$, $Y$, and $Z'$ as Figure~\ref{fig trie}, where open (resp. solid) bullets indicate matches $w$ with $i_w < i$ (resp. $i_w \geq i$), double-edged bullets indicate matches $(\mathit{Suff}_i[j], j)$, and the matches $(\mathit{Suff}_i[j], j)$ composing $\mathit{ext}(Z')$ are connected by a dotted line
}
\label{fig D221}
\end{figure}

If $\mathit{D221}$ is available, then each element $v$ in $\mathit{ext}(Z')$ can be obtained in $O(1)$ in descending order of $j_v$ inductively as stated below.

\begin{lemma}\label{lem D221}
Let $Z'$ be an arbitrary MCS-prefix and let $i = i_{\mathit{pref}(Z')} + 1$.
Let $j_1 \circ j_2 \circ \cdots \circ j_\ell$ be the sequence of indices such that $j_1 = \mathit{RMQ}_{\mathit{Suff}_i}(j_{\mathit{pref}(Z')} + 1 : \min(j_{\mathit{safe}(Z')}, |Y|))$, $j_r = \mathit{RMQ}_{\mathit{Suff}_i}(j_{\mathit{pref}(Z')} + 1 : j_{r - 1} - 1)$ for any index $r$ with $2 \leq r \leq \ell$, $\mathit{Suff}_i[j_\ell] \leq \min(i_{\mathit{safe}(Z')}, |X|)$, and either $j_\ell = j_{\mathit{pref}(Z')} + 1$ or $\mathit{Suff}_i[\mathit{RMQ}_{\mathit{Suff}_i}(j_{\mathit{pref}(Z')} + 1 : j_\ell - 1)] > \min(i_{\mathit{safe}(Z')}, |X|)$.
Sequence $\mathit{ext}(Z')$ consists of suff-matches $(\mathit{Suff}_i[j_r], j_r)$ for all indices $r$ with $1 \leq r \leq \ell$.
\end{lemma}

\begin{proof}
Let $\mathit{ext}(Z')$ consist of prominent witnesses $v_1,v_2,\dots,v_{|\mathit{ext}(Z')|}$ of $Z'$-extensibility with $j_{v_1} > j_{v_2} > \cdots > j_{v_{|\mathit{ext}(Z')|}}$.
Let $r$ be an arbitrary index with $1 \leq r \leq |\mathit{ext}(Z')|$.
Since $v_r$ is prominent, $\mathit{Suff}_i[j_{v_r}] = i_{v_r}$ and there exists no suff-match $v$ such that $j_{\mathit{pref}(Z')} + 1 \leq j_v \leq j_{v_r} - 1$ and $i_{\mathit{pref}(Z')} + 1 \leq i_v \leq i_{v_r}$.
Furthermore, if $r = 1$, then there exists no suff-match $v$ such that $j_{v_1} < j_v \leq \min(j_{\mathit{safe}(Z')}, |Y|)$ and $i_{\mathit{pref}(Z')} + 1 \leq i_v < i_{v_1}$; otherwise, there exists no suff-match $v$ such that $j_{v_r} + 1 \leq j_v \leq j_{v_{r + 1}} - 1$ and $i_{\mathit{pref}(Z')} + 1 \leq i_v < i_{v_r}$.
Therefore, it can be proven by induction that $j_{v_r} = j_r$.
Since $i_{v_{|\mathit{ext}(Z')|}} \leq \min(i_{\mathit{safe}(Z')}, |X|)$ and there exists no suff-match $v$ such that $j_{\mathit{pref}(Z')} + 1 \leq j_v \leq j_{v_{|\mathit{ext}(Z')|}} - 1$ and $i_{\mathit{pref}(Z')} + 1 \leq i_v \leq \min(i_{\mathit{safe}(Z')}, |X|)$, $|\mathit{ext}(Z')| = \ell$.
\end{proof}

According to Lemma~\ref{lem D221}, we can obtain $\mathit{ext}(Z')$ in $O(|\mathit{ext}(Z')|)$ time using $\mathit{D221}$.
Thus, we have the following corollary of the lemma.

\begin{corollary}\label{cor D221}
$\mathit{D221}$ supports $O(|\mathit{ext}(Z')|)$-time queries of $\mathit{ext}(Z')$ for any MCS-prefix $Z'$.
\end{corollary}

\begin{figure}
\centering
\begin{algorithm}
\aitem{1}{$\mathit{Suff}_i \leftarrow \mathit{Suff}_{i + 1}$;}
\aitem{1}{for each index $j$ from $1$ to $|Y|$,}
\aitem{2}{if $\mathit{Suff}_{i + 1}[j] \neq |X| + 1$, $\mathit{prev}_Y(X[i], j) \geq 1$, and \\ $\mathit{prev}_X(X[i], \mathit{Suff}_{i + 1}[j]) = i$, then}
\aitem{3}{$\mathit{Suff}_i[\mathit{prev}_Y(X[i], j)] \leftarrow i$;}
\aitem{1}{output $\mathit{Suff}_i$.}
\end{algorithm}
\caption{
Procedure $\mathsf{UpdateSuff}(i, \mathit{Suff}_{i + 1})$}
\label{algo UpdateSuff}
\end{figure}

Algorithm $\mathsf{Enum221}$ prepares data structure $\mathit{D221}$ by constructing sequence $\mathit{Suff}_i$ inductively for each index $i$ from $|X| - 1$ to $1$ in descending order using Procedure $\mathsf{UpdateSuff}(i, \mathit{Suff}_{i + 1})$ presented in Figure~\ref{algo UpdateSuff}, where the initial sequence $\mathit{Suff}_{|X|} = (|X| + 1) \circ (|X| + 1) \circ \cdots \circ (|X| + 1) \circ |X|$ is constructed from scratch.
Once each $\mathit{Suff}_i$ is obtained, $\mathit{RMQ}_{\mathit{Suff}_i}$ can be constructed in $O(n)$ time \cite{BF}.

\begin{lemma}\label{lem Suff}
For any index $i$ with $1 \leq i \leq |X| - 1$, if $\mathit{Table}_{\mathit{next/prev}}$ and $\mathit{Suff}_{i + 1}$ are available, then Procedure $\mathsf{UpdateSuff}(i, \mathit{Suff}_{i + 1})$ outputs $\mathit{Suff}_i$ in $O(n)$ time.
\end{lemma}

\begin{proof}
For any index $j$ with $1 \leq j \leq |Y|$, if $(i, j)$ is a suff-match, then $\mathit{Suff}_i[j] = i$; otherwise, $\mathit{Suff}_i[j] = \mathit{Suff}_{i + 1}[j]$.
Furthermore, for any index $j'$ with $1 \leq j' \leq |Y|$, $(i, j')$ is a suff-match if and only if there exists an index $j$ with $j' < j \leq |Y|$ such that $(\mathit{Suff}_{i + 1}[j], j)$ is a match (i.e., $\mathit{Suff}_{i + 1}[j] \neq |X| + 1$), $\mathit{prev}_Y(X[i], j) = j'$, and $\mathit{prev}_X(X[i], \mathit{Suff}_{i + 1}[j]) = i$.
Thus, the procedure determines $\mathit{Suff}_i$ correctly.
Since $\mathit{Table}_{\mathit{next/prev}}$ supports $O(1)$-time next/prev-queries, the procedure runs in $O(n)$ time.
\end{proof}

Lemma~\ref{lem Suff} immediately implies that the following corollary holds.

\begin{corollary}\label{cor Suff}
$\mathit{D221}$ can be constructed in $O(n^2)$ time.
\end{corollary}

Consequently, from Corollaries~\ref{cor Enum221}, \ref{cor D221}, and \ref{cor Suff}, we have the following theorem.

\begin{theorem}\label{theo Enum221} 
Algorithm $\mathsf{Enum221}$ is an $(n^2, n^2, n)$-algorithm that solves the MCS enumeration problem.
\end{theorem}

\subsection{$O(n^2, n, n \log n)$-algorithm}\label{sec Enum211}

As another implementation of Algorithm $\mathsf{Basic}$ based on Conte~et~al.~\cite{CGP+}'s prefix-extensible character test (Lemmas~\ref{lem ext} and \ref{lem safe}), we propose an $O(n^2, n, n \log n)$-algorithm that solves the MCS enumeration problem, which we denote Algorithm $\mathsf{Enum211}$.

For any MCS-prefix $Z'$, let $\mathit{step}(Z')$ denote the minimum of $a_{\mathit{pref}(Z' \circ c)} - a_{\mathit{pref}(Z')}$ over all $Z'$-extensible characters $c$. 
Algorithm $\mathsf{Enum211}$ is almost the same as Algorithm $\mathsf{Enum221}$.
The only difference is that Algorithm $\mathsf{Enum211}$ uses $\mathit{Index}_{\mathit{next/prev}}$ for supporting $O(\log n)$-time next/prev-queries and a data structure $\mathit{D211}$ of size $O(n)$ that supports $O(\mathit{step}(Z') \log n)$-time queries of $\mathit{ext}(Z')$.
From this, we immediately have a lemma corresponding to Corollary~\ref{cor Enum221} as follows.

\begin{lemma}\label{lem Enum211}
If $\mathit{D211}$ is of size $O(n)$ and supports $O(\mathit{step}(Z') \log n)$-time queries of $\mathit{ext}(Z')$ for any MCS-prefix $Z'$, then Algorithm $\mathsf{Enum211}$ uses $O(n)$ space, including space for storing $\mathit{Index}_{\mathit{next/prev}}$ and $\mathit{D211}$, to output all distinct MCSs of $X$ and $Y$ one by one each in $O(n \log n)$ time.
\end{lemma}

Below we develop data structure $\mathit{D211}$, which is of size $O(n)$ and supports $O(\mathit{step}(Z') \log n)$-queries of $\mathit{ext}(Z')$ for any MCS-prefix $Z'$.
Our idea for achieving the design of such a data structure is to classify all matches into $O(n)$ types based on their ``character-wise diagonal coordinates'' and determine a certain unique suff-match for each type to support $O(1)$-time queries of whether $w$ is a suff-match for any match $w$.
For any match $w$, let $\hat{\imath}_w$ (resp. $\hat{\jmath}_w$) denote the number of indices $i$ (resp. $j$) with $1 \leq i \leq i_w$ (resp. $1 \leq j \leq j_w$) such that $X[i] = c_w$ (resp. $Y[j] = c_w$), so that $i_w$ (resp. $j_w$) is the $\hat{\imath}_w$th (resp. $\hat{\jmath}_w$th) least such index $i$ (resp. $j$).
Furthermore, let $\hat{d}_w$ (resp. $\hat{a}_w$) denote $\hat{\jmath}_w - \hat{\imath}_w$ (resp $\hat{\imath}_w + \hat{\jmath}_w$), which is called the \emph{character-wise diagonal} (resp. \emph{anti-diagonal}) \emph{coordinate} of $w$.
The key observation is stated in the following lemma, which claims that there exists a threshold with respect to the character-wise anti-diagonal coordinate that separates the same type matches into suff-matches and others. 

\begin{lemma}\label{lem monotone}
For any suff-match $w$ and any match $v$ such that $c_v = c_w$, $\hat{d}_v = \hat{d}_w$, and $\hat{a}_v < \hat{a}_w$, $v$ is also a suff-match.
\end{lemma}

\begin{proof}
For any string $Z''$ such that $\mathit{suff}(Z'') = w$, $v = \mathit{suff}({c_w}^{(\hat{a}_v - \hat{a}_w) / 2} \circ Z'')$, where ${c_w}^{(\hat{a}_v - \hat{a}_w) / 2}$ is the string consisting only of $(\hat{a}_v - \hat{a}_w) / 2$ copies of $c_w$.
\end{proof}

For any character $c$ and any index $\hat{d}$, let $\mathit{th}(c, \hat{d})$ denote the suff-match $w$ with $c_w = c$ and $\hat{d}_w = \hat{d}$ that has the greatest $\hat{a}_w$, if any, or a virtual suff-match $(\max(0, -\hat{d}), \max(0, \hat{d}))$, otherwise (see Figure~\ref{fig D211}). 
Lemma~\ref{lem monotone} implies that we can determine whether any match $w$ is a suff-match only from $\hat{\imath}_{w^\star}$ (resp. $\hat{\jmath}_{w^\star}$), where $w^\star = \mathit{th}(c_w, \hat{d}_w)$.
That is, $w$ is a suff-match if and only if $\hat{a}_w \leq \hat{a}_{w^\star}$, which holds if and only if $\hat{\imath}_w \leq \hat{\imath}_{w^\star}$ (resp. $\hat{\jmath}_w \leq \hat{\jmath}_{w^\star}$) due to $\hat{d}_w = \hat{d}_{w^\star}$.
Based on this observation, we define $\mathit{D211}$ as follows.

\begin{figure}[t]
\centering
\includegraphics[width=10cm]{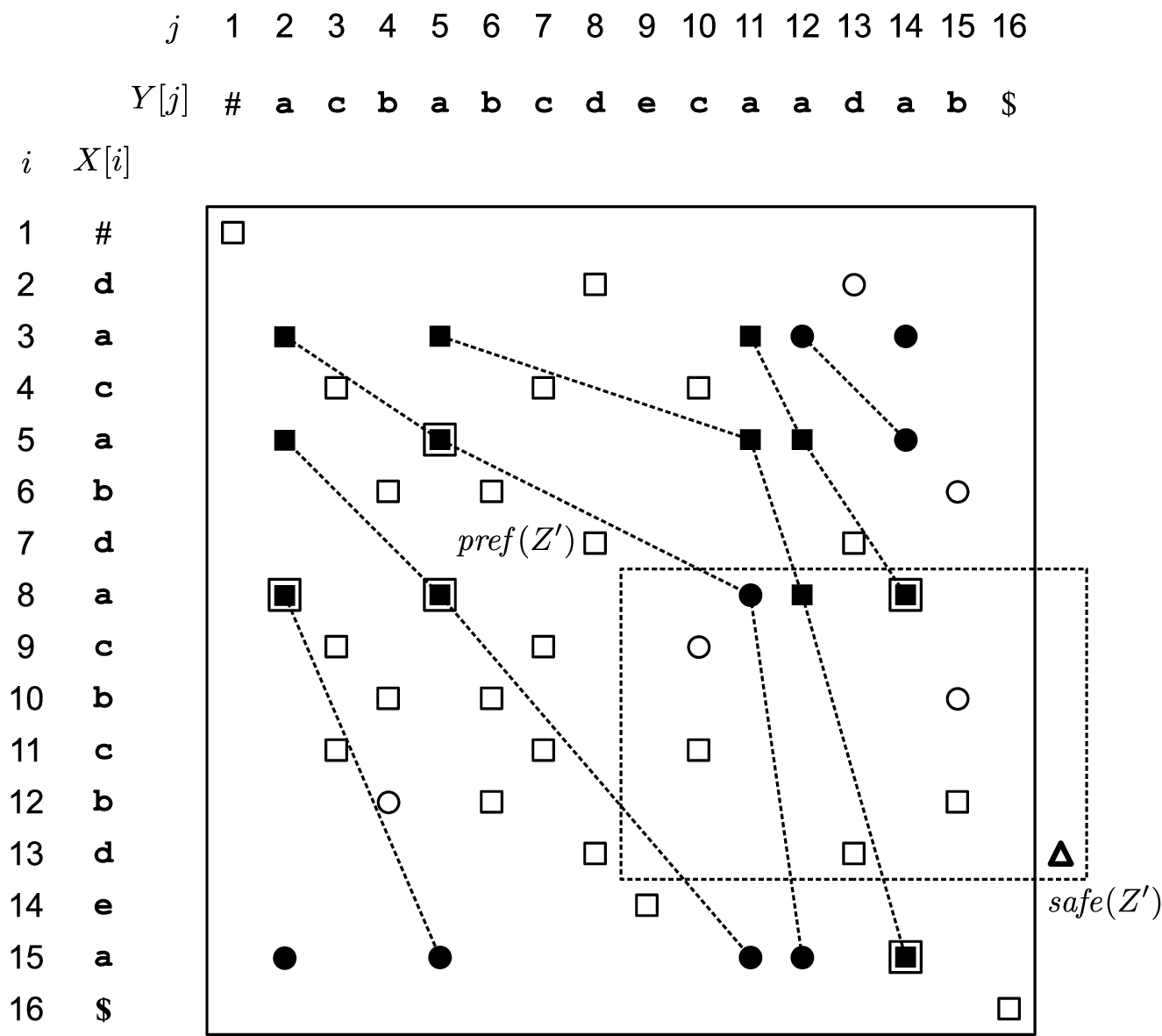}
\caption{
Suff-matches $\mathit{th}(\mathtt{a}, \hat{d})$, excluding virtual ones, with $-3 \leq \hat{d} \leq 4$ (double edged), which defines $\hat{I}^{\mathit{th}}_{\mathtt{a}} = 3 \circ 3 \circ 3 \circ 2 \circ 4 \circ 3 \circ 0 \circ 0$ and $\hat{J}^{\mathit{th}}_{\mathtt{a}} = 0 \circ 1 \circ 2 \circ 2 \circ 5 \circ 5 \circ 3 \circ 4$, for the same $X$, $Y$, and $Z'$ as Figure~\ref{fig trie}, where solid (resp. open) bullets indicate matches $w$ with $c_w = \mathtt{a}$ (resp. $c_w \neq \mathtt{a}$) and each line or polygonal line connects matches $w$ with $c_w = \mathtt{a}$ having the same character-wise diagonal coordinate $\hat{d}_w$
}
\label{fig D211}
\end{figure}

\begin{definition}\label{def D211}
For any character $c$ with $1 \leq c \leq \sigma$, let $\#_{X, c}$ (resp. $\#_{Y, c}$) denote the number of indices $i$ (resp. $j$) with $1 \leq i \leq |X|$ (resp. $1 \leq j \leq |Y|$) such that $X[i] = c$ (resp. $Y[j] = c$).
Furthermore, let $\hat{I}^{\mathit{th}}_c$ (resp. $\hat{J}^{\mathit{th}}_c$) denote the sequence of $\#_{X, c} + \#_{Y, c} - 1$ indices such that $\hat{I}^{\mathit{th}}_c[\hat{d} + \#_{X, c}] = \hat{\imath}_{w^\star}$ (resp. $\hat{J}^{\mathit{th}}_c[\hat{d} + \#_{X, c}] = \hat{\jmath}_{w^\star}$) for any index $\hat{d}$ with $1 - \#_{X, c} \leq \hat{d} \leq \#_{Y, c} - 1$, where $w^\star = \mathit{th}(c, \hat{d})$.
$\mathit{D211}$ consists of sequences $\hat{I}^{\mathit{th}}_c$ and $\hat{J}^{\mathit{th}}_c$ and their range maximum query data structures $\mathit{RMQ}_{-\hat{I}^{\mathit{th}}_c}$ and $\mathit{RMQ}_{-\hat{J}^{\mathit{th}}_c}$ for all characters $c$ with $1 \leq c \leq \sigma$.
(As auxiliary data, $\mathit{D211}$ also contains sequence $\hat{I}$ of length $|X|$ (resp. $\hat{J}$ of length $|Y|$), sequence $\#_X$ (resp. $\#_Y$) of length $\sigma$, and sequences $I_c$ of length $\#_{X, c}$ (resp. $J_c$ of length $\#_{Y, c}$) for all characters $c$ with $1 \leq c \leq \sigma$ such that for any match $w$, $\hat{I}[i_w] = \hat{\imath}_w$ (resp. $\hat{J}[j_w] = \hat{\jmath}_w$), $\#_X[c_w] = \#_{X, c_w}$ (resp. $\#_Y[c_w] = \#_{Y, c_w}$), and $I_{c_w}[\hat{\imath}_w] = i_w$ (resp. $J_{c_w}[\hat{\jmath}_w] = j_w$), which all can be prepared by a single scan of $X$ (resp. $Y$) in $O(n)$ time and $O(n)$ space.)
\end{definition}

For any MCS-prefix $Z'$, let any character $c$ such that $X[i_{\mathit{pref}(Z')} + 1 \rangle$ and $Y[j_{\mathit{pref}(Z')} + 1 \rangle$ share $c$ and $\mathit{pref}(Z') \prec \mathit{pref}(Z' \circ c)$ be called a \emph{$Z'$-extensible character candidate}.
Hence, for any match $w$ with $\mathit{pref}(Z') < w$, if $\mathit{pref}(Z') \prec w$ does not hold, then there exists a $Z'$-extensible character candidate $c$ such that $\mathit{pref}(Z' \circ c) < w$.
This implies that any $Z'$-extensible character is a $Z'$-extensible character candidate.
For any $Z'$-extensible character candidate $c$, there exist at most two candidates of prominent $c$-witnesses of $Z'$-extensibility.
One candidate is the suff-match $v_\leftarrow$ such that $i_{v_\leftarrow} = i_{\mathit{pref}(Z' \circ c)}$ and $j_{v_\leftarrow}$ is the least possible $j$ with $j_{\mathit{pref}(Z' \circ c)} \leq j \leq j_{\mathit{safe}(Z')}$ and $Y[j] = c$, if any.
The other candidate is the suff-match $v_\uparrow$ such that $j_{v_\uparrow} = j_{\mathit{pref}(Z' \circ c)}$ and $i_{v_\uparrow}$ is the least possible $i$ with $i_{\mathit{pref}(Z' \circ c)} \leq i \leq i_{\mathit{safe}(Z')}$ and $X[i] = c$, if any.
We call these at most two existing candidates the \emph{prominent $c$-witness candidates of $Z'$-extensibility}.
Note that for any such candidate $v$, $\mathit{pref}(Z') \prec v$ if and only if $v$ is a prominent $c$-witness of $Z'$-extensibility.
We use $\mathit{D211}$ to search for the prominent $c$-witness candidates of $Z'$-extensibility in $O(\log n)$ time for each $Z'$-extensible character candidates $c$ based on the following lemma.

\begin{lemma}\label{lem v-cand}
For any MCS-prefix $Z'$ and any $Z'$-extensible character candidate $c$, if $\mathit{Index}_{\mathit{next/prev}}$ is available, then $\mathit{D211}$ can be used to search for any of the prominent $c$-witness candidates of $Z'$-extensibility in $O(\log n)$ time.
\end{lemma}

\begin{proof}
Since the candidate $v_\uparrow$ with $j_{v_\uparrow} = j_{\mathit{pref}(Z' \circ c)}$ can be searched for in a symmetric manner, we show how to search for the candidate $v_\leftarrow$ with $i_{v_\leftarrow} = i_{\mathit{pref}(Z' \circ c)}$ in $O(\log n)$ in time.

Let $w_\vdash = \mathit{pref}(Z' \circ c)$ and let $w_\dashv$ be the match $(i_{\mathit{pref}(Z' \circ c)}, j_\dashv)$, where $j_\dashv$ is the greatest index with $j_{\mathit{pref}(Z' \circ c)} \leq j_\dashv \leq j_{\mathit{safe}(Z')}$ such that $Y[j_\dashv] = c$.
These two matches can be determined in $O(\log n)$ time using $\mathit{Index}_{\mathit{next/prev}}$.
Since $v_\leftarrow$ is the suff-match $v$ having the least possible $j_v$ such that $i_v = i_{w_\vdash}$ and $j_{w_\vdash} \leq j_v \leq j_{w_\dashv}$, $v_\leftarrow$ exists if and only if there exists a suff-match $v$ such that $i_v = i_{w_\vdash}$ and $j_{w_\vdash} \leq j_v \leq j_{w_\dashv}$.
Such a suff-match $v$ exists if and only if there exists a character-wise diagonal coordinate $\hat{d}$ with $\hat{d}_{w_\vdash} \leq \hat{d} \leq \hat{d}_{w_\dashv}$ such that $\hat{\imath}_{w_\vdash} \leq \hat{\imath}_{\mathit{th}(c, \hat{d})}$.
Furthermore, such a $\hat{d}$ exists if and only if 
\[
\hat{\imath}_{w_\vdash} \leq \hat{I}^{\mathit{th}}_c[\mathit{RMQ}_{-\hat{I}^{\mathit{th}}_c}(\hat{d}_{w_\vdash} + \#_{X, c} : \hat{d}_{w_\dashv} + \#_{X, c})].
\]
Thus, $\mathit{D211}$ can be used to determine whether $v_\leftarrow$ exists in $O(1)$ time.
If $v_\leftarrow$ exists, then $\hat{\jmath}_{w_\vdash} \leq \hat{\jmath}_{v_\leftarrow} \leq \hat{\jmath}_{w_\dashv}$.
Hence, we can use $\mathit{D211}$ to determine $\hat{\jmath}_{v_\leftarrow}$ in $O(\log n)$ time by a binary search based on the fact that for any indices $\hat{\jmath}$, $\hat{\jmath}'$, and $\hat{\jmath}''$ with $\hat{\jmath}_{w_\vdash} \leq \hat{\jmath} \leq \hat{\jmath}' \leq \hat{\jmath}'' \leq \hat{\jmath}_{w_\dashv}$ such that $\hat{\jmath} \leq \hat{\jmath}_{v_\leftarrow} \leq \hat{\jmath}''$, if 
\[
\hat{\imath}_{w_\vdash} \leq \hat{I}^{\mathit{th}}_c[\mathit{RMQ}_{-\hat{I}^{\mathit{th}}_c}((\hat{\jmath} - \hat{\imath}_{w_\vdash}) + \#_{X, c} : (\hat{\jmath}' - \hat{\imath}_{w_\vdash}) + \#_{X, c})],
\]
then $\hat{\jmath} \leq \hat{\jmath}_{v_\leftarrow} \leq \hat{\jmath}'$; otherwise, $\hat{\jmath}' + 1 \leq \hat{\jmath}_{v_\leftarrow} \leq \hat{\jmath}''$.
\end{proof}

As stated in the following lemma, the number of $Z'$-extensible character candidates is appropriately small for our purpose, and all are found efficiently.

\begin{lemma}\label{lem c-cand}
For any MCS-prefix $Z'$, the number of $Z'$-extensible character candidates is less than $\mathit{step}(Z')$.
If $\mathit{Index}_{\mathit{next/prev}}$ is available, then the sequence of matches $\mathit{pref}(Z' \circ c)$ for all $Z'$-extensible character candidates $c$ in ascending order of $j_{\mathit{pref}(Z' \circ c)}$ (hence also in descending order of $i_{\mathit{pref}(Z' \circ c)}$) can be determined in $O(\mathit{step}(Z') \log n)$ time. 
\end{lemma}

\begin{proof}
Let $S = (i_1, j_1) \circ (i_2, j_2) \circ \cdots \circ (i_s, j_s)$ be the sequence of matches in the lemma.
By an argument similar to the proof of Lemma~\ref{lem Enum221}, for any index $r$ with $1 \leq r \leq s$, $a_{(i_r, j_r)} - a_{\mathit{pref}(Z')} \geq s + 1$.
This implies that $s < \mathit{step}(Z')$.

To construct $S$, we repeatedly extend either prefix $S \langle p]$ to $S \langle p + 1]$, if $a_{(i_p, j_p)} \geq a_{(i_q, j_q)}$, or suffix $S[q \rangle$ to $S[q - 1 \rangle$, otherwise, until $S[p] = S[q]$.
Let $(i_t, j_t)$ be determined to extend either $S \langle q - 1]$ to $S \langle q]$ or $S[p + 1 \rangle$ to $S[p \rangle$ in the last iteration.
Due to the condition as to whether $S \langle p]$ or $S[q \rangle$ should be extended, $a_{(i_t, j_t)} - a_{\mathit{pref}(Z')} = \mathit{step}(Z')$.
For any index $p$ with $0 \leq p \leq t - 1$, $(i_{p + 1}, j_{p + 1})$ can be determined in $O((j_{p + 1} - j_p) \log n)$ time by finding $j_{p + 1}$, which is the least index $j$ with $j_p < j$ such that $\mathit{next}_X(Y[j], i_{\mathit{pref}(Z')}) < i_p$, and setting $i_{p + 1}$ to $\mathit{next}_X(Y[j_{p + 1}], i_{\mathit{pref}(Z')})$, where $i_0 = |X| + 1$ and $j_0 = j_{\mathit{pref}(Z')}$.
Therefore, it takes $O((j_t - j_{\mathit{pref}(Z')}) \log n)$ time to obtain $S \langle t]$.
Analogously, it takes $O((i_t - i_{\mathit{pref}(Z')}) \log n)$ time to obtain $S[t \rangle$, completing the proof.
\end{proof}

From Lemmas~\ref{lem v-cand} and \ref{lem c-cand}, we can obtain $\mathit{ext}(Z')$ for any MCS-prefx $Z'$ in $O(\mathit{step}(Z') \log n)$ time as follows.

\begin{lemma}\label{lem D211}
For any MCS-prefix $Z'$, if $\mathit{Index}_{\mathit{next/prev}}$ and $\mathit{D211}$ are available, then $\mathit{ext}(Z')$ can be obtained in $O(\mathit{step}(Z') \log n)$ time.
\end{lemma}

\begin{proof}
Let $S = (i_1, j_1) \circ (i_2, j_2) \circ \cdots \circ (i_s, j_s)$ be the same sequence of matches as in the proof of Lemma~\ref{lem c-cand}, which is hence of length at most $\mathit{step}(Z')$ and obtained in $O(\mathit{step}(Z') \log n)$ time.
Let $r$ be an arbitrary index with $1 \leq r \leq s$.
Let $v$ be an arbitrary prominent $c_{(i_r, j_r)}$-witness candidate of $Z'$-extensibility, which can be obtained in $O(\log n)$ time by Lemma~\ref{lem v-cand}.
There exists an index $p$ (resp. $q$) with $1 \leq p \leq r - 1$ (resp. $r + 1 \leq q \leq s$) such that $(i_p, j_p) < v$ (resp. $(i_q, j_q) < v)$ if and only if $(i_{r - 1}, j_{r - 1}) < v$ (resp. $(i_{r + 1}, j_{r + 1}) < v$).
Thus, whether $v$ is a witness of $Z'$-extensibility can be determined in $O(1)$ time.
\end{proof}

\begin{figure}
\centering
\begin{algorithm}
\aitem{1}{$\mathit{Suff}_i \leftarrow \mathit{Suff}_{i + 1}$;}
\aitem{1}{$i_{\mathit{next}} \leftarrow \mathit{next}_X(X[i], i)$, which is  determined by scanning $X[i + 1 \rangle$;}
\aitem{1}{$j_{\mathit{prev}} \leftarrow 0$;}
\aitem{1}{for each index $j$ from $1$ to $|Y|$,}
\aitem{2}{if $\mathit{Suff}_{i + 1}[j] \neq |X| + 1$, $j_{\mathit{prev}} \geq 1$, and $\mathit{Suff}_{i + 1}[j] \leq i_{\mathit{next}}$, then}
\aitem{3}{$\mathit{Suff}_i[j_{\mathit{prev}}] \leftarrow i$;}
\aitem{2}{if $Y[j] = X[i]$, then}
\aitem{3}{$j_{\mathit{prev}} \leftarrow j$;}
\aitem{1}{output $\mathit{Suff}_i$.}
\end{algorithm}
\caption{
Procedure $\mathsf{UpdateSuff2}(i, \mathit{Suff}_{i + 1})$}
\label{algo UpdateSuff2}
\end{figure}

We construct $\mathit{D211}$ in a straightforward manner by enumerating all suff-matches.
To do this in $O(n^2)$ time and $O(n)$ space, we modify Procedure $\mathsf{UpdateSuff}$ so as to run in $O(n)$ time without using $\mathit{Table}_{\mathit{next/prev}}$ as follows.

\begin{lemma}\label{lem UpdateSuff2}
$\mathit{D211}$ can be constructed in $O(n^2)$ time and $O(n)$ space.
\end{lemma}

\begin{proof}
We initialize all elements $\hat{I}^{\mathit{th}}_c[\hat{d} + \#_{X, c}]$ (resp. $\hat{J}^{\mathit{th}}_c[\hat{d} + \#_{X, c}]$) of $\hat{I}^{\mathit{th}}_c$ (resp. $\hat{J}^{\mathit{th}}_c$) to $\max(0, -\hat{d})$ (resp. $\max(0, \hat{d})$).
This can be done in $O(n)$ time.
Then, for each suff-match $v$, if $\hat{I}^{\mathit{th}}_{c_v}[\hat{d}_v + \#_{X, c_v}] < \hat{\imath}_v$ (resp. $\hat{J}^{\mathit{th}}_{c_v}[\hat{d}_v + \#_{X, c_v}] < \hat{\jmath}_v$), then the value of $\hat{I}^{\mathit{th}}_{c_v}[\hat{d}_v + \#_{X, c_v}]$ (resp. $\hat{J}^{\mathit{th}}_{c_v}[\hat{d}_v + \#_{X, c_v}]$) is updated to $\hat{\imath}_v$ (resp. $\hat{\jmath}_v$).
All suff-matches are enumerated by inductively constructing $\mathit{Suff}_i$ for each index $i$ with $1 \leq i \leq |X| - 1$ in descending order after $\mathit{Suff}_{|X|}$ is constructed from scratch.
We use Procedure $\mathsf{UpdateSuff2}(i, \mathit{Suff}_{i + 1})$ in Figure~\ref{algo UpdateSuff2} to obtain $\mathit{Suff}_i$ from $\mathit{Suff}_{i + 1}$.
This procedure uses variables $i_{\mathit{next}}$ and $j_{\mathit{prev}}$ to maintain indices $\mathit{next}_X(X[i], i)$ and $\mathit{prev}_Y(X[i], j)$, respectively.
Variable $i_{\mathit{next}}$ is determined by line 2 of the procedure in $O(n)$ time while variable $j_{\mathit{prev}}$ is maintained dynamically according to the value of variable $j$ by lines 3, 7, and 8.
Since $(\mathit{Suff}_{i + 1}[j], j)$ is a match and $\mathit{prev}_X(X[i], \mathit{Suff}_{i + 1}[j]) = i$ if and only if $\mathit{Suff}_{i + 1}[j] \neq |X| + 1$ and $\mathit{Suff}_{i + 1}[j] \leq i_{\mathit{next}}$, we can prove that Procedure $\mathsf{UpdateSuff2}(i, \mathit{Suff}_{i + 1})$ outputs $\mathit{Suff}_i$ in $O(n)$ time by the same argument as the proof of Lemma~\ref{lem Suff}.
Once $\mathit{Suff}_i$ is obtained, all suff-matches $w$ with $i_w = i$ can be extracted from it in $O(n)$ time.
Thus, all suff-matches can be enumerated in $O(n^2)$ time and $O(n)$ space.
\end{proof}

Lemmas~\ref{lem Enum211}, \ref{lem D211}, and \ref{lem UpdateSuff2} immediately yield the following theorem.

\begin{theorem}\label{theo Enum211} 
Algorithm $\mathsf{Enum211}$ is an $(n^2, n, n \log n)$-algorithm that solves the MCS enumeration problem.
\end{theorem}

\section{Conclusion}\label{sec conc}

This article considered the problem of enumerating maximal common subsequences (MCSs) of two strings. 
For any positive integer $n$ and any pair of strings both of length $O(n)$, $(n^3, n^3, n)$-, $(n^2, n^2, n)$-, and $(n^2, n, n \log n)$-algorithms for this problem were proposed, where an $(f_{\mathrm{p}}(n), f_{\mathrm{s}}(n), f_{\mathrm{d}}(n))$-algorithm outputs all distinct MCSs of the two strings each in $O(f_{\mathrm{d}}(n))$ time after performing an $O(f_{\mathrm{p}}(n))$-time preprocessing to prepare a data structure of size $O(f_{\mathrm{s}}(n))$.
Although the $(n^3, n^3, n)$-time algorithm is inferior to the $(n^2, n^2, n)$-algorithm in terms of efficiency, the data structure constructed by the $(n^3, n^3, n)$-time algorithm allows access to only all distinct MCSs without explicitly enumerating them, so it can be used to efficiently find or enumerate certain special MCSs, such as quasi-LCSs and most stable MCSs.

An interesting question remains as to whether it is possible to efficiently enumerate MCSs for an arbitrary number of strings.
This is because a naive generalization of any of the algorithms proposed in this article would result in a data structure that is exponential in size with respect to the number of strings.
If MCSs of multiple strings are enumerable with the preprocessing-time, space, and delay-time complexities polynomial in the number and length of the strings, then we can treat all common subsequences, including LCSs, which are NP-hard to find, as included in the search scope for significant structures shared by the strings.

\section*{Acknowledgment}

This work was supported by JSPS KAKENHI Grant Number JP23K10975.

\end{document}